\documentclass[preprint,3p]{elsarticle_arxiv}

\usepackage{graphicx}

\usepackage{color}

\usepackage{amssymb}
\usepackage{amsthm}

\usepackage{amsmath,latexsym,epsfig,wrapfig}
\usepackage{amsfonts}%

\usepackage{xspace}
\usepackage[noend]{algorithmic}
\usepackage{algorithm}
\usepackage{setspace}

\newtheorem{theorem}{Theorem}
\newtheorem{lemma}[theorem]{Lemma}
\newtheorem{corollary}[theorem]{Corollary}
\newtheorem{proposition}[theorem]{Proposition}
\theoremstyle{definition}

\theoremstyle{remark}
\newtheorem{remark}[theorem]{Remark}
\newtheorem{example}{Example}

\newcommand{\DeltaT}{\Delta_f(t)}
\newcommand{\frtwo}{\tfrac{1}{\sqrt{2}}}
\newcommand{\rtwo}{\sqrt{2}}

\newcommand{\rrone}{\rule[.7ex]{2.5em}{.8pt}}
\newcommand{\rrtwo}{\raisebox{0pt}[0pt]{\rule[.7ex]{1em}{0pt}\rule[.7ex]{.8pt}{2.8ex}\rule[.7ex]{1.5em}{.8pt}}}

\makeatletter

\newcommand{\Rmnum}[1]{\expandafter\@slowromancap\romannumeral #1@}
\makeatother

\newcommand{\casequadricsquadrics}{\Rmnum{1}\ }
\newcommand{\casequadricsplanes}{\Rmnum{2}\ }
\newcommand{\casequadricsconics}{\Rmnum{3}\ }
\newcommand{\caseplanesconics}{\Rmnum{4}\ }
\newcommand{\casequadricslines}{\Rmnum{5}\ }
\newcommand{\casequadricsvertices}{\Rmnum{6}\ }
\newcommand{\caseconicsconicsthree}{\Rmnum{7}\ }
\newcommand{\caseconicsconicstwo}{\Rmnum{8}\ }
\newcommand{\caseconicslinesthree}{\Rmnum{9}\ }
\newcommand{\caseconicslinestwo}{\Rmnum{10}\ }
\newcommand{\caselineslines}{\Rmnum{11}\ }

\newcommand{\raisenumber}[3]{\hspace{#1}\raisebox{#2}{#3}}
\newcommand{\raisenumbertwo}[3]{\hspace{#1} \raisebox{#2}{#3}}
\newcommand{\boundnum}[1]{\hspace{-10pt}\rule[#1]{11pt}{0.3pt} \hspace{-5.3pt}
\rule[#1]{0.3pt}{9.5pt} }
\newcommand{\Lower}[1]{\smash{\lower 1.5ex \hbox{#1}}}
\newcommand{\MLower}[1]{\smash{\lower 4.5ex \hbox{#1}}}
\newcommand{\TLower}[1]{\smash{\lower 2.5ex \hbox{#1}}}

\begin{document}

\begin{frontmatter}



\title{Continuous Collision Detection for Composite Quadric Models}


\author[hku]{Yi-King Choi\corref{cor1}}
\ead{ykchoi@cs.hku.hk}

\author[hku]{Wenping Wang}
\ead{wenping@cs.hku.hk}

\author[inria]{Bernard Mourrain}
\ead{mourrain@sophia.inria.fr}

\author[sdu]{Changhe Tu}
\ead{chtu@sdu.edu.cn}

\author[cas]{Xiaohong Jia}
\ead{xhjia@amss.ac.cn}

\author[hku]{Feng Sun}
\ead{fsun@cs.hku.hk}

\cortext[cor1]{Corresponding author}
\address[hku]{The University of Hong Kong, Pokfulam Road, Hong Kong, China}
\address[inria]{GALAAD, INRIA M\'{e}diterran\'{e}e, 2004 route des lucioles,
                06902 Sophia-Antipolis, France}
\address[sdu]{Shandong University, Jinan, China}
\address[cas]{KLMM, AMSS \& NCMIS, Chinese Academy of Science, Beijing, China}

\begin{abstract}
A composite quadric model (CQM) is an object modeled by 
piecewise linear or quadric patches.
We study the continuous detection problem of a special type of CQM objects which are commonly used in CAD/CAM, that is, the boundary surfaces of such a CQM intersect only in straight line segments or conic curve segments. 
We present a framework for continuous collision detection (CCD) of
this special type of CQM (which we also call CQM for brevity) in motion.
We derive algebraic formulations and compute numerically the first contact time instants
and the contact points of two moving CQMs in $\mathbb R^3$.
Since it is difficult to process CCD of two CQMs in a direct manner because they are composed of semi-algebraic varieties, 
we break down the problem into subproblems of solving CCD of pairs of 
boundary elements of the CQMs. 
We present procedures to solve CCD
of different types of boundary element pairs in different dimensions.
Some CCD problems are reduced to their equivalents in a lower dimensional setting,
where they can be solved more efficiently.
\end{abstract}


\end{frontmatter}

\section{Introduction}
Collision detection is important to many fields involving
object interaction and simulation, e.g., computer animation, computational physics,
virtual reality, robotics, CAD/CAM and virtual manufacturing.
Its primary purpose is to determine possible contacts or intersections
between objects so that proper responses may be further
carried out accordingly.
There has been considerable research in relation to collision detection, particularly
in the field of robotics and computer graphics, regarding the different issues such
as intersection tests, bounding volume computation, graphics hardware speedup, etc.
Among these studies, continuous collision detection (CCD) is currently
an active research topic, in which collision status within a continuous time span is determined.

Quadric surfaces form an important class of objects used in practice.
In CAD/CAM or industrial manufacturing, objects are often designed and
modeled using quadric surfaces because of their simple representations and ease of
handling. Quadric surfaces encompass all degree two surfaces, which include the
commonly used spheres, ellipsoids, cylinders and cones. Ellipsoids, truncated/capped
cylinders and cones are usually used as approximations to complex geometry in graphics
and robotics.  Furthermore, most mechanical parts can be modeled accurately with quadric
surfaces.  Through composite representation or CSG (constructive solid geometry) 
composition, an even wider class of complex objects are modeled by quadric surfaces.

Most existing collision detection methods are intended for piecewise linear objects 
such as triangles, boxes, polyhedrons, or simple curved primitives such as spheres.
Collision detection of objects containing quadric surfaces may be done by applying
these methods to piecewise linear approximations of the objects.
This, however, introduces geometric error and entails large storage space.
As a result, exact collision detection of quadric surfaces is important due to the extensive use of quadric surfaces as modeling primitives in applications.

In this paper we present a framework for efficient and exact
continuous collision detection (CCD) of \emph{composite quadric models}, or CQMs for short.
CQMs are modeled by piecewise linear or quadric surface patches.
The \emph{boundary elements} of a CQM may either be 
a \emph{face} (a linear or quadric surface patch), an \emph{edge} (where two faces meet) 
or a \emph{vertex} (where three or more edges meet). 
A boundary edge of a CQM is in general a degree four intersection curve of two quadrics.
However, there is a special class of CQMs whose boundary edges 
are straight line or conic curve segments only (Figure~\ref{fig:cqm_obj_eg}).
In this paper, we focus on CCD of this special class of CQMs (which we shall also denote 
as ``CQM'' for brevity),
which is by itself an important problem due to the popular use of the class in practice.
This work also represents a step towards tackling CCD of general CQMs,
which is difficult to be solved efficiently.

\begin{figure}[!ht]
\centering
    \includegraphics[width=.5\linewidth]{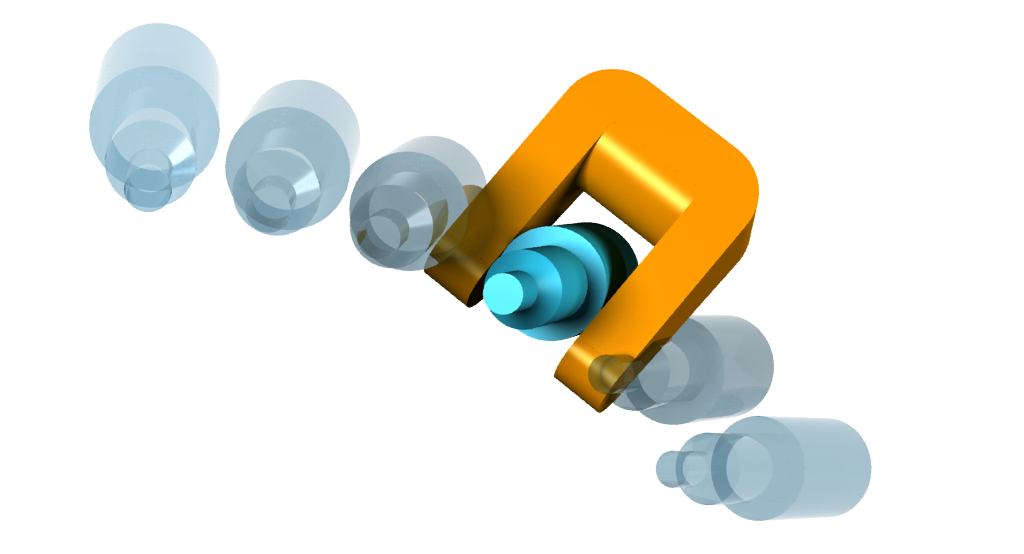}
\caption{Two CQMs in motion.  The objects are typical examples of 
the special class of CQMs whose
boundary edges are straight line or conics curve segments only.}
\label{fig:cqm_obj_eg}
\end{figure}

Our main contributions are as follows.
\begin{itemize}
\item We present a framework for exact and efficient continuous collision detection (CCD)
      of two moving composite quadric models (CQMs).
      Given two moving CQMs which are separate
      initially, our method computes their first contact time and contact point.
      The CQMs may undergo both the Euclidean and affine motions, 
      which means that the objects may either be rigid or change their shapes
      under affine transformations. 
\item Our framework comprises a collection of algebraic methods for CCD
      of different types of boundary components of a CQM.  In particular, 
      \begin{enumerate}
      \item we devise an algorithm for CCD of two moving quadrics (Section~\ref{sect:ccdquadrics}), which is based on our recent result of detecting morphological change of intersection curve for two moving quadrics (\cite{JiaWC2012}); and 
      \item we derive algebraic conditions for different configurations of 1D conics in $\mathbb{PR}$ and further devise an algorithm for CCD of two moving conics in 3D (Section~\ref{sect:CCD_space_conics}).
      \end{enumerate}
\end{itemize}

\section{Related Work}

\subsection{Continuous collision detection}

Different approaches have been proposed for solving continuous collision detection (CCD)
for various types of moving objects.
There are CCD methods by equation solving, which include
\cite{Canny1986} and \cite{RedonKC2000} for polyhedra,
\cite{ChoiWLK2006} for elliptic disks and
\cite{ChoiCWK2008} for ellipsoids.

Swept volumes (SV) are also commonly used:
\cite{Cameron1990} presents a solution using a four-dimensional space-time SV;
\cite{KimRLMT2007} and \cite{RedonLMK2005}
deal with CCD of articulated bodies by considering
SVs of line swept spheres (LSS); and
\cite{GovindarajuKLM2006} works on SVs of triangles to solve CCD of deformable models
with significant speedup using GPU.

Efficiency and accuracy are the major concerns for CCD. 
\cite{ZhangRLK2007} use the approach
of conservative advancement and achieve acceleration of CCD for articulated objects by using the Taylor model which is a generalization of interval arithmetic. 
For deforming triangle meshes, \cite{MinKM2010} proposes conservative local advancement that significantly improve CCD performance by computing motion bounds for the bounding volumes of the primitives.
A recent work by \cite{BrochuEB2012} uses geometrically exact predicates for efficient and accurate CCD of deforming triangle meshes.

Our CCD method works on exact representations of CQM models and is based on algebraic formulations.  For better efficiency, equation solving for obtaining contact time instants and contact points is done numerically.

\subsection{Intersection and collision of quadrics}

Classifications and computations of the intersections of two 
general quadrics are thoroughly studied in classical algebraic
geometry (\cite{Bromwich1906,SempleK1952,Berger1987})
and CAGD (\cite{Wang2002, Dupont2004, DupontLLP2008a,DupontLLP2008b,DupontLLP2008c,TuWW2002, TuWMW2009}).
These results, however, consider quadrics
in the complex or real projective space, and
are not applicable to collision detection problems which concern only
the real affine or Euclidean space.
There is nevertheless an obvious way to detect intersection
between stationary quadrics by computing their
real intersection curves.
Various algorithms have also been proposed
(e.g.,~\cite{Levin1979,Miller1987, WilfM1993, WangJG2002, WangGT2003}),
whose objectives are
to classify the topological or geometric structure of
the intersection curves and to derive their parametric representations.
However, these methods are difficult to extend for 
collision detection of moving quadrics.

Our previous work in \cite{ChoiWLK2006,ChoiCWK2008} presents
algorithms for exact CCD of elliptic disks and ellipsoids,
based on an algebraic
condition for the separation of two ellipsoids established by~\cite{WangWK2001}.
Although quadrics are widely used in many applications,
CCD of general quadrics has not been addressed in the
literature.
We propose recently an algebraic method for detecting the morphological change of the intersection curves of two moving quadrics in 3D real projective space (\cite{JiaWC2012}).
In this paper, we further devise an algorithm for CCD of moving quadrics which is a subproblem of CCD of CQMs.
We also develop a framework to solve CCD of CQMs, the more general class of objects
composed of piecewise linear or quadric primitives.

\section{Outline of Algorithm}

Two moving CQMs ${\cal Q}_A(t)$  and ${\cal Q}_B(t)$,
 where $t$ is a time parameter in the interval
$[t_0, t_1]$, are said to be {\em collision-free}, if the
intersection of ${\cal Q}_A(t)$ and ${\cal Q}_B(t)$ is empty for all
$t \in [t_0, t_1]$; otherwise, they are said to {\em collide}.
Two CQMs ${\cal Q}_A(t')$  and ${\cal Q}_B(t')$ at a particular time instant $t'$ are {\em in contact} or {\em touching},
if their boundaries have nonempty intersection while their interiors are disjoint. 
Given two initially separate CQMs,
our goal is to determine whether the CQMs are collision-free or not;
if they collide, their first contact time instant in $[t_0, t_1]$
and the contact point will be computed.

We assume that the CQMs undergo
artbitrary motions which are expressible as continuous functions of the time
parameter $t$.
Among the various motion types,
rational motions are easily handled by CAGD techniques 
\label{res:CAGD_tools}
that deal with splines and polynomials.
Our method involves root finding, and in the case of rational motions, the functions are 
polynomials whose roots can be efficiently solved for by these techniques.
Moreover, low-degree rational motions are found to be sufficient for modeling
smooth motions in most applications and hence further enhance efficiency.
See~\cite{JuettlerW2002} for a thorough discussion of rational motion design.
While rational motions are used in our examples, our method is also applicable
to other motions, such as helical motions which are transcendental.
Numerical solver will then be needed for root finding of these functions.

CQMs can be viewed as semi-algebraic varieties which are defined by multiple
polynomial inequalities.  Their boundary elements are often 
finite pieces on a quadric or a conic and hence it is
difficult to process CQMs using algebraic methods in a direct
manner.
To tackle CCD of CQMs, we consider pairwise CCD between
the {\em extended boundary elements} (Figure~\ref{fig:extended_element}) which are defined as follows:
\begin{itemize}
\item The complete planar or quadric surface containing a boundary face 
of a CQM $Q$ is called an {\em extended boundary face} of $Q$.
\item The complete straight line or conic curve containing a boundary edge
of a CQM $Q$ is called the {\em extended boundary edge} of $Q$.
\item An {\em extended boundary element} of a CQM $Q$ is either an extended boundary face, an extended boundary edge, or a vertex of $Q$.
\end{itemize}
\begin{figure}[!ht]
\centering
  \begin{minipage}{.5\linewidth}
    \includegraphics[width=\linewidth]{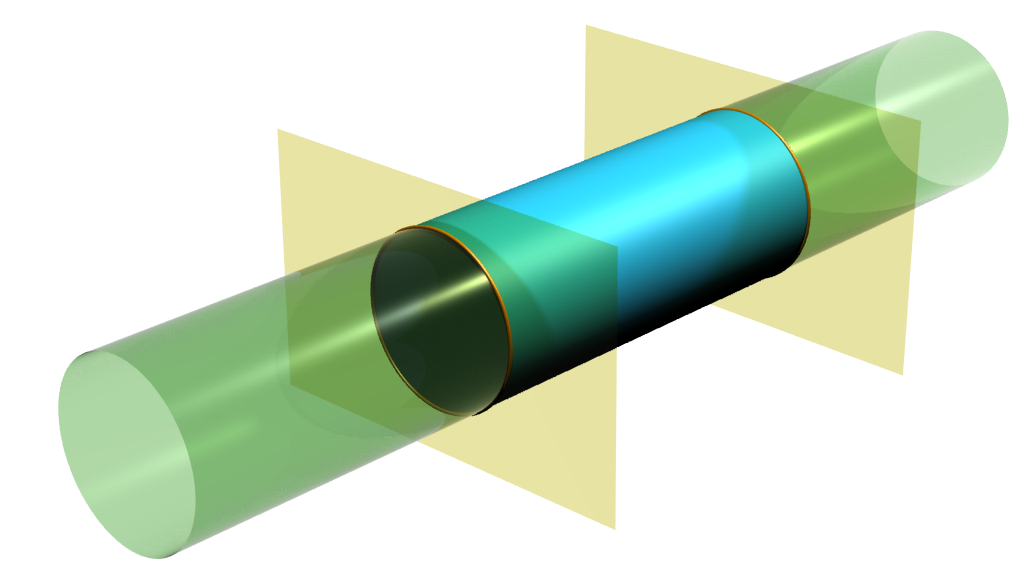}
  \end{minipage}
\caption{A capped cylinder (in blue) and its extended boundary elements.
The cylinder (in green) and the two planes (in yellow) are the extended boundary elements of the cylindrical surface
and the two disks of the capped ends, respectively.  The circular edge (in orange) is the extended boundary element of
itself.}
\label{fig:extended_element}
\end{figure}

It follows that CCD of CQMs entails solving CCD of different element types.
For example, to detect possible
contact between two moving capped cylinders (Figure~\ref{capped cylinders}), 
one should handle CCD of
(a) cylinder vs. cylinder; (b) cylinder vs. ellipse; (c) ellipse vs. plane;
and (d) ellipse vs. ellipse.

\label{res:outline}
The major steps of our algorithm are outlined as follows:
\begin{enumerate}
\item Given two CQMs, we first identify CCD subproblems between all
possible pairs of their extended boundary elements (Section~\ref{sect:subproblems}).
\item For each CCD subproblem, we use an
algebraic method to compute their first contact instant and point of contact.
We will present a classification of different types of CCD
problems that one may encounter in CCD of CQMs and discuss the detailed solution to each case (Section~\ref{sect:solveCCD}).
\item Once a contact is found between two extended boundary elements, we will
check if the contact is valid, that is, if it lies on both CQMs (Section~\ref{sect:contact_validation}),
since a contact found in Step 2 may lie on a portion of an extended boundary element that
is not part of a CQM boundary element.  
Two CQMs are in contact only if the contact point between
the extended elements lies on both CQMs (see Figure~\ref{fig:valid_contact}).
\item After the CCD subproblems are solved, the first valid contact among all pairs of
boundary elements is then the first contact of the two CQMs.
\end{enumerate}

\begin{figure}[!ht]
\centering
  \begin{minipage}{.8\linewidth}
    \begin{minipage}{.45\linewidth}
       \centering
       \includegraphics[width=\linewidth]{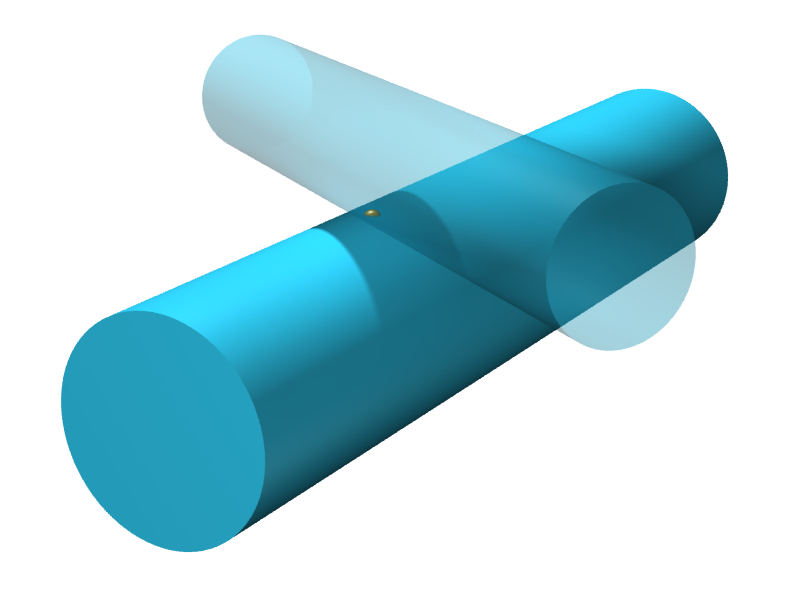}\\
       {\scriptsize (a)}
    \end{minipage}
    \hfill
    \begin{minipage}{.45\linewidth}
       \centering
       \includegraphics[width=\linewidth]{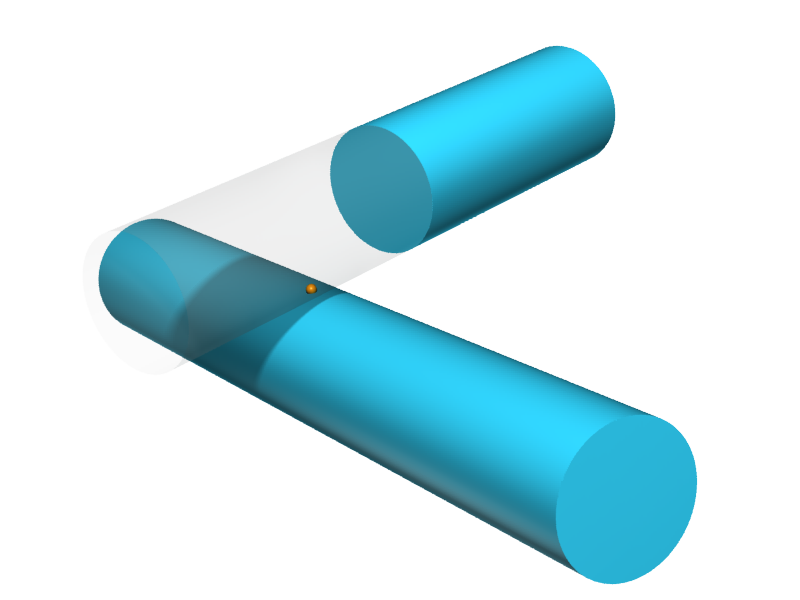}\\
       {\scriptsize (b)}
    \end{minipage}
  \end{minipage}
\caption{Contact validation.  (a) The extended boundary elements (cylinders) have a valid contact point that lie on both CQMs, so the capped cylinders are in contact.  (b) The contact point of the extended boundary elements is an invalid contact of the CQMs since it does not lie on both CQMs.}
\label{fig:valid_contact}
\end{figure}

The main algorithm is given in Algorithm~\ref{alg:main} as follows.  

\begin{center}
\begin{minipage}{0.8\textwidth}
\begin{algorithm}[H]
\caption{The Main Algorithm}\label{alg:main}
\begin{algorithmic}
\REQUIRE Two moving CQMs ${\mathcal Q}_A(t)$ and ${\mathcal Q}_B(t)$, $t \in [t_0, t_1]$, and ${\mathcal Q}_A(t_0) \cap {\mathcal Q}_B(t_0) = \emptyset$ 
\ENSURE Whether ${\mathcal Q}_A(t)$ and ${\mathcal Q}_B(t)$ are collision-free or colliding, and the first contact time and contact point in case of collision
\medskip
\STATE Identify all CCD subproblems between the extended boundary elements of
       ${\mathcal Q}_A(t)$ and ${\mathcal Q}_B(t)$.
\FOR{each CCD subproblem}
  \STATE Find, if there is any, the first candidate contact time $t_i$ with a valid
         contact point ${\bf p}_i$ that lies on both ${\mathcal Q}_A(t_i)$ and ${\mathcal Q}_B(t_i)$.
  \STATE ${\mathcal S} \leftarrow {\mathcal S} \cup \{(t_i, \mathbf{p}_i)\}$
\ENDFOR
\IF {${\mathcal S} = \emptyset$}
  \RETURN ${\mathcal Q}_A(t)$ and ${\mathcal Q}_B(t)$ is collision-free for $t \in [t_0, t_1]$
\ELSE
  \STATE $i^* \leftarrow \arg \min \{ t_i \mid (t_i, \mathbf{p}_i) \in \mathcal{S} \}$
  \RETURN $(t_{i^*}, \mathbf{p}_{i^*})$ as the first contact time and contact point of ${\mathcal Q}_A(t)$ and ${\mathcal Q}_B(t)$
\ENDIF
\end{algorithmic}
\end{algorithm}
\end{minipage}
\end{center}

\section{Identifying subproblems}
\label{sect:subproblems}

A contact of two CQMs always happens between a pair of boundary elements, one from each of the
CQMs.
Depending on the types of the boundary elements, we have different types of 
contacts---$(F,F)$, $(F,E)$,
$(F,V)$, $(E,E)$, $(E,V)$ and $(V,V)$, where $F,E$ and $V$ stand for face, edge and
vertex, respectively.
Let $\tilde I$ denote the extended boundary element of a boundary element $I$ of a CQM.
The contact types can be characterized by the geometric configuration of the extended boundary elements of two CQMS as follows:
\begin{description}
\item[$(F,F)$-type:] If a contact of type $(F,F)$ happens between two boundary faces $F_1$ and $F_2$, then the two extended boundary faces $\tilde F_1$ and $\tilde F_2$ either have a tangential contact or are identical.
\item[$(F,E)$-type:] If a contact of type $(F,E)$ happens between a boundary face $F_1$ and a boundary edge $E_2$, then either the extended boundary face $\tilde F_1$ and and the extended boundary edge $\tilde E_2$ are in tangential contact, or $\tilde E_2$ is contained in $\tilde F_1$.
\item[$(F,V)$-type:] If a contact of type $(F,V)$ happens between a boundary face $F_1$ and a vertex $V_2$, then $V_2$ must lie on the extended boundary face $\tilde F_1$.
\item[$(E,E)$-type:] If a contact of type $(E,E)$ happens between two boundary edges $E_1$ and $E_2$, then
the extended boundary edges $\tilde E_1$ and $\tilde E_2$ either have a real intersection or are identical.
\item[$(E,V)$-type:] If a contact of type $(E,V)$ happens between a boundary edge $E_1$ and a vertex $V_2$, then
$V_2$ must lie on the extended boundary edge $\tilde E_1$.
\item[$(V,V)$-type:] If a contact of type $(V,V)$ happens between two vertices $V_1$ and $V_2$, then $V_1$ and $V_2$ must coincide.
\end{description}

We now show that one only needs to deal with four of the contact types for CCD of two CQMs.

\begin{proposition}
A contact between two CQMs can be classified into one of the four basic types:  $(F,F)$, $(F,E)$, $(F,V)$ and $(E,E)$.
\end{proposition}
\begin{proof}
A contact between two CQMs can be of more than one contact type, since it may lie on two or more extended boundary elements.  
Both the $(E,V)$- and $(V,V)$-type contacts can be treated as $(F,V)$-type.
In particular, an $(E,V)$-type contact between a boundary edge $E_1$ and a vertex $V_2$
is always also an $(F,V)$-type contact between $V_2$ and a boundary face on which $E_1$ lies.
A $(V,V)$-type contact between two vertices $V_1$ and $V_2$ is always also an 
$(F,V)$-type contact between $V_2$ and a boundary face on which $V_1$ lies.
\end{proof}

A CCD subproblem is defined for each pair of extended boundary elements, one from each of the two given CQMs.  
As the detection of $(F,V)$-type contacts would capture all the $(E,V)$- and $(V,V)$-type contacts, we have the following:

\begin{corollary}
It suffices to consider CCD subproblems of the four basic contact types---$(F,F)$, $(F,E)$, $(F,V)$ and $(E,E)$ to solve CCD of two CQMs.
\end{corollary}

\begin{example}\label{eg:cqm_concave}
Figure~\ref{fig:CQM_objects} shows two CQMs which is (a) a ring
constructed by subtracting a cylinder from an ellipsoid; and (b) 
a wedge formed by first subtracting three half spaces and then a circular cylinder
from an elliptic cylinder.
The CCD subproblems for the two CQMs are as follows:
\begin{itemize}
\item 8 $(F,F)$-type subproblems: $\{ F_{{\mathcal A},i} \}_{i=1}^2 \times \{ F_{{\mathcal B},j} \}_{j=1}^4$.
\item 20 $(F,E)$-type subproblems: $\{ F_{{\mathcal A},i} \}_{i=1}^2 \times \{ E_{{\mathcal B},j} \}_{j=1}^6$
              and
              $\{ E_{{\mathcal A},i} \}_{i=1}^2 \times \{ F_{{\mathcal B},j} \}_{j=1}^4$.
\item 8 $(F,V)$-type subproblems: $\{ F_{{\mathcal A},i} \}_{i=1}^2 \times \{ V_{{\mathcal B},j} \}_{j=1}^4$.
\item 12 $(E,E)$-type subproblems: $\{ E_{{\mathcal A},i} \}_{i=1}^2 \times \{ E_{{\mathcal B},j} \}_{j=1}^6$.
\end{itemize}
\qed
\end{example}

\begin{figure}[!ht]
\centering
  \begin{minipage}{.8\linewidth}
  \begin{minipage}{.45\linewidth}\centering
    \includegraphics[width=\linewidth]{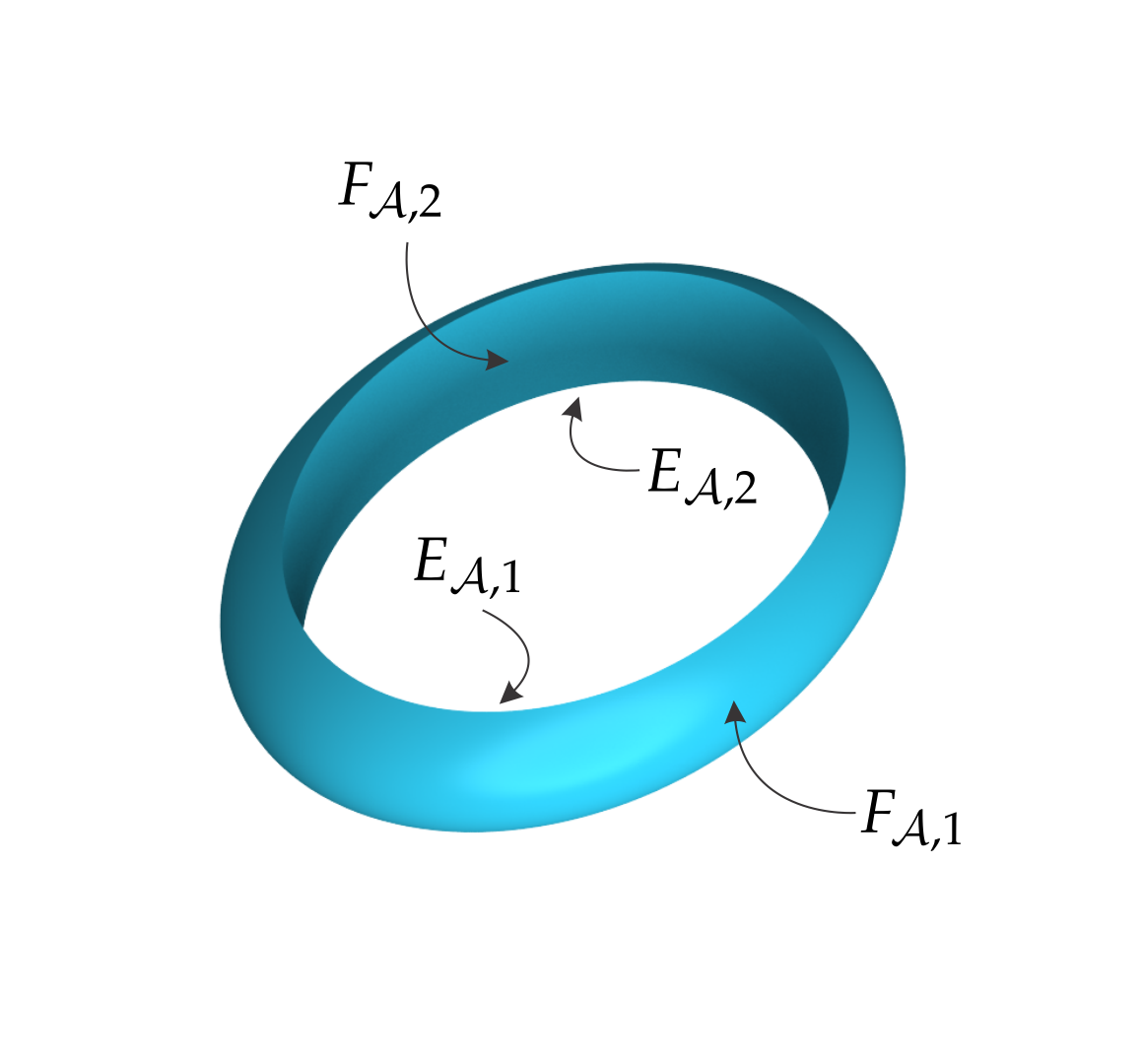}\\ \vspace{-1ex}
    {\scriptsize (a)}
  \end{minipage}\hfill
  \begin{minipage}{.45\linewidth}\centering
    \includegraphics[width=\linewidth]{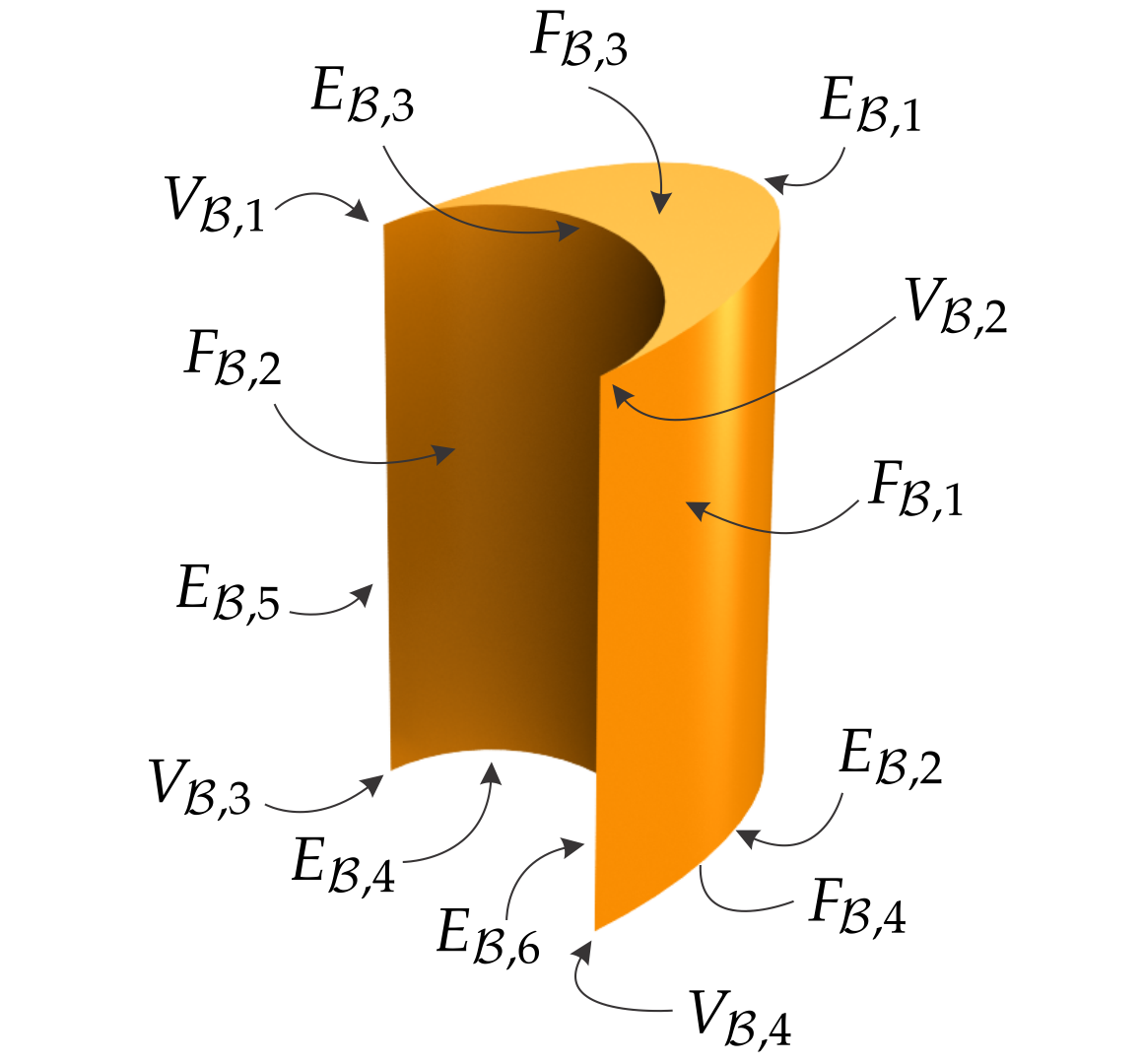}\\ \vspace{-1ex}
    {\scriptsize (b)}
  \end{minipage}
  \end{minipage}
\caption{Different types of boundary elements of two CQMs, where
$F,E$ and $V$ stand for face, edge and vertex, respectively.}
\label{fig:CQM_objects}
\end{figure}

\section{Solving CCD subproblems}\label{sect:solveCCD}

For a pair of extended boundary elements, 
each from CQMs ${\mathcal Q}_A(t)$ and ${\mathcal Q}_B(t)$,
respectively, the next step
is to solve their CCD and compute the first contact time instant with the
corresponding contact point.
There is a hierarchy of extended boundary elements, from faces to
vertices, in different dimensions.  Each element type also consists
of more than one kind of primitives; for instance, a face may either be 
a quadric face or be a planar face.
A complete classification of the types of element pairs that should be
considered for CCD of two CQMs is listed in Table~\ref{tab:summary}.
In this section, we shall present the techniques for resolving
CCD of these cases.

\begin{table}
\centering
\caption{Complete classification of different types of element pairs of two CQMs
and the technique for solving the corresponding CCD.  
Note that CCD between two planes and CCD between a plane and a line can be exempted and therefore are not listed here.}
\label{tab:summary}
\begin{tabular}{rcccc}
\hline
Type \rule{2.8em}{0pt}& Case & Element pairs & Techniques & Section \\ \hline
$(F,F)$ \rrone
        & \casequadricsquadrics & Quadrics vs. Quadrics & CCD of quadrics & \ref{sect:ccdquadrics}\\
        \rrtwo
        & \casequadricsplanes & Quadrics vs. Planes &  CCD of quadrics/planes & \ref{sect:ccd_quadrics_planes}\\
$(F,E)$ \rrone
        & \casequadricsconics & Quadrics vs. Conics & Dimension reduction to Case~\caseconicsconicstwo & \ref{sect:CCD_quadrics_conics}\\
        \rrtwo
        & \caseplanesconics & Planes vs. Conics & Dimension reduction to Case~\caseconicslinestwo & \ref{sect:CCD_planes_conics}\\
        \rrtwo
        & \casequadricslines & Quadrics vs. Lines & Direct substitution & \ref{sect:CCD_quadrics_lines}\\
$(F,V)$ \rrone
        & \casequadricsvertices & Quadrics/Planes vs. Vertices & Direct substitution & \ref{sect:CCD_quadrics_vertices}\\
$(E,E)$ \rrone
        & \caseconicsconicsthree & Conics vs. Conics in $\mathbb{R}^3$& Dimension reduction & \ref{sect:CCD_space_conics}\\
		\rrtwo
		& \caseconicsconicstwo & Conics vs. Conics in $\mathbb{R}^2$& CCD of conics in ${\mathbb R}^2$ &
\ref{sect:CCD_conics_2D}\\
        \rrtwo
        & \caseconicslinesthree & Conics vs. Lines in $\mathbb{R}^3$& Dimension reduction & \ref{sect:CCD_conics_lines_3D}\\
        \rrtwo
        & \caseconicslinestwo & Conics vs. Lines in $\mathbb{R}^2$& Direct substitution & \ref{sect:CCD_conics_lines_2D}\\
        \rrtwo 
        & \caselineslines & Lines vs. Lines & CCD of linear primitives & \ref{sect:CCD_lines}\\
\hline
\end{tabular}
\end{table}

\label{res:plane_line}
We note here that CCD between two planes can be exempted since any planar face 
of a CQM must be delimited by some boundary curves and any possible contact of two planes can
be found by CCD between one planar face and a boundary curve of another.  
Similarly, CCD between a plane and a line can also be exempted, since any possible contact between a planar
face and a straight edge of two CQMs can be found by CCD between the boundary curve of the face and the line, or CCD between a boundary vertex of the line and the plane.
Hence, CCD for these two cases are not listed in Table~\ref{tab:summary}.

\subsection{Case~\casequadricsquadrics --- quadrics vs. quadrics}\label{sect:ccdquadrics}

In this section, we deal with CCD of two quadric surfaces in ${\mathbb R}^3$.
We assume that the quadrics are irreducible and hence they do not represent planes.
CCD between a quadric surface and a plane is discussed in Section~\ref{sect:ccd_quadrics_planes}.

Given two moving quadric surfaces in ${\mathbb R}^3$, our goal is to compute the time instants at which there is a contact between the two quadrics.  
There are three different local contact configurations between two quadrics:
surface contact, curve contact or point contact (Fig.~\ref{fig:quadrics_contact}).
%
Two quadrics have a {\em surface contact} if and only if they are identical. 
They have a {\em curve contact} if and only if they are tangent at every point along a line or conic curve.
There is a {\em point contact} if and only if they are tangent at an {\em isolated} common point.
It is also important that the quadrics do not intersect locally at the neighbourhood of all the tangent points.
Figure~\ref{fig:cylinder_intersect}(a) shows two cylinders that are tangent at a point but also intersect locally at the neigbourhood of the same point.  The tangent point therefore does not constitute a contact.  

\begin{figure}[!ht]
\centering
  \begin{minipage}{\linewidth}
  \begin{minipage}{.32\linewidth}\centering
    \includegraphics[width=\linewidth]{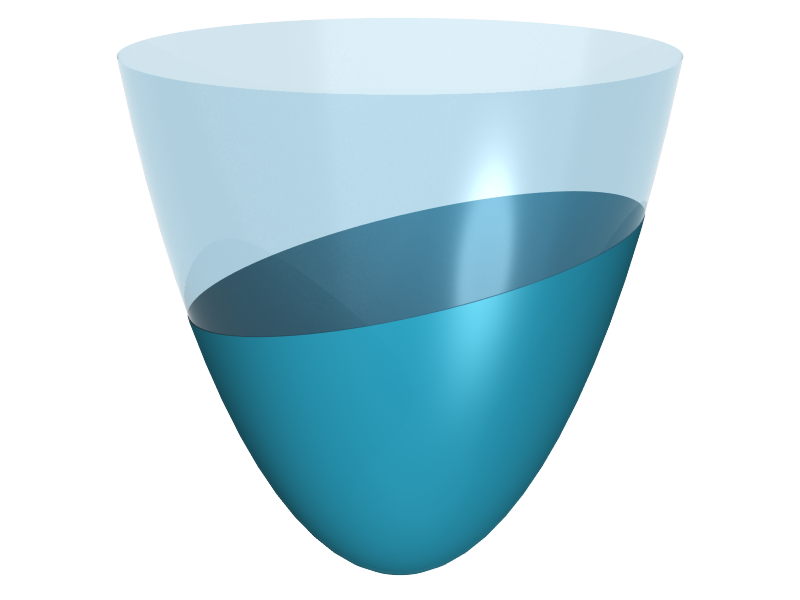}\\ \vspace{-1ex}
    {\scriptsize (a)}
  \end{minipage}\hfill
  \begin{minipage}{.32\linewidth}\centering
    \includegraphics[width=\linewidth]{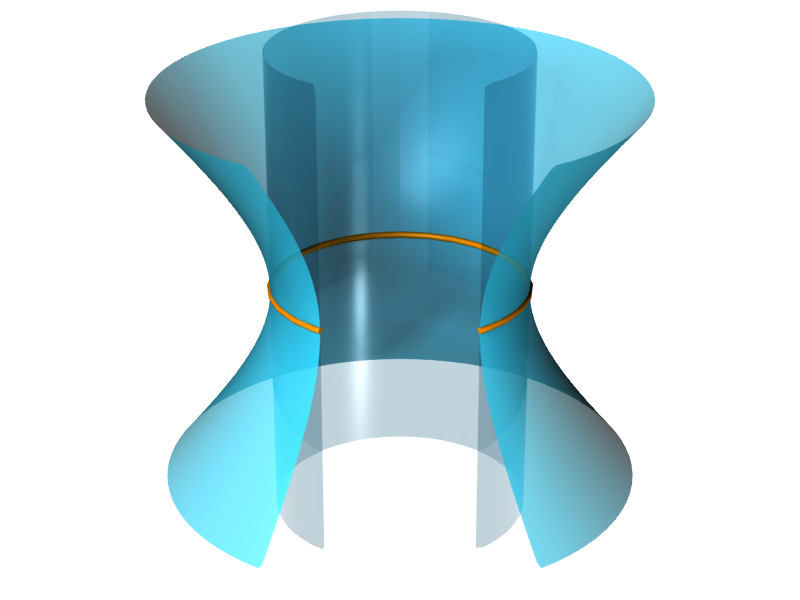}\\ \vspace{-1ex}
    {\scriptsize (b)}
  \end{minipage}\hfill
  \begin{minipage}{.32\linewidth}\centering
    \includegraphics[width=\linewidth]{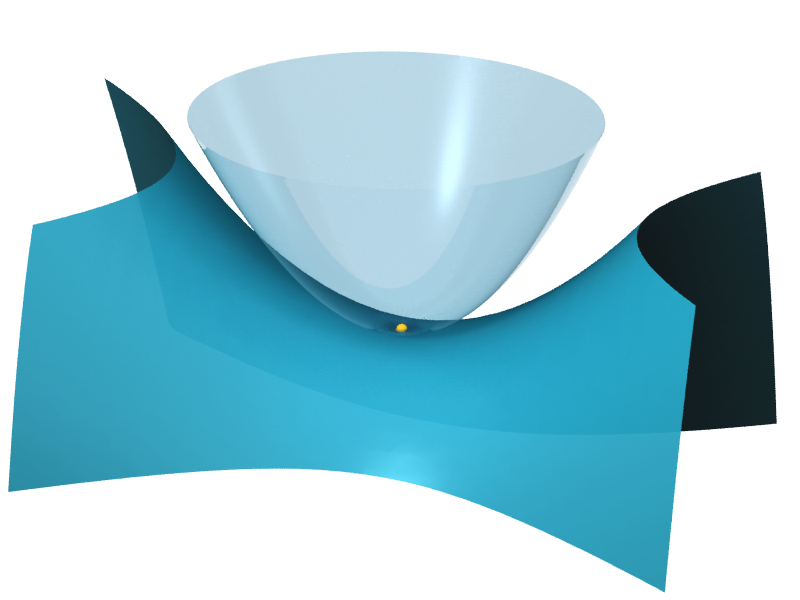}\\ \vspace{-1ex}
    {\scriptsize (c)}
  \end{minipage}
  \end{minipage}
\caption{The three different local contact configurations between two quadrics: (a) surface contact; (b) curve contact; and (c) point contact.}
\label{fig:quadrics_contact}
\end{figure}

\begin{figure}[!ht]
\centering
\begin{minipage}{.7\linewidth}
\centering
\begin{minipage}{0.48\linewidth}
\centering
\includegraphics[width=\linewidth]{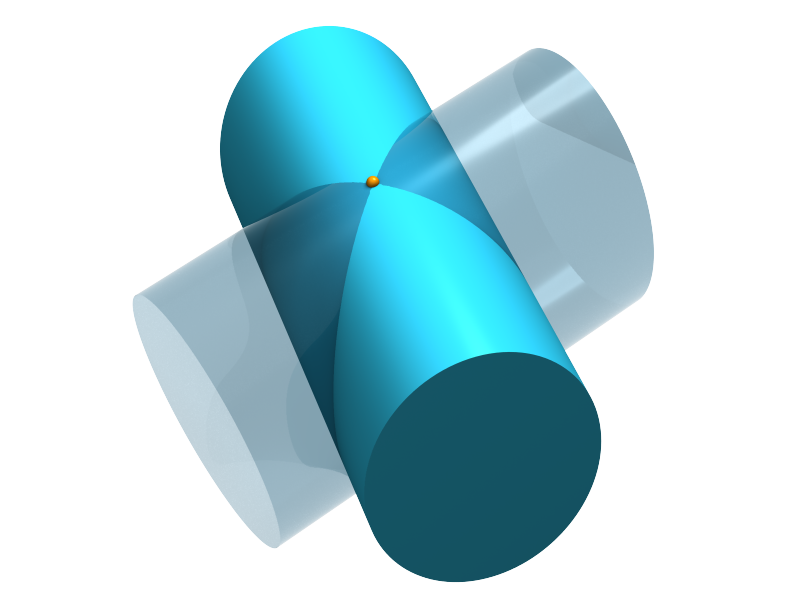}\\
{\small (a)}
\end{minipage}
\hfill
\begin{minipage}{0.48\linewidth}
\centering
\includegraphics[width=\linewidth]{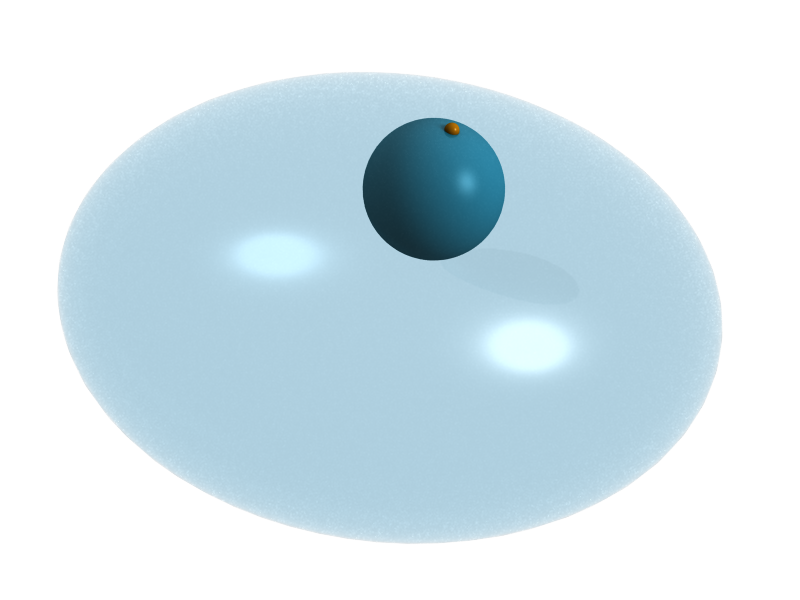}\\
{\small (b)}
\end{minipage}
\end{minipage}
\caption{(a) Two cylinders intersect locally at the neighbourhood of the tangent point.
(b) There is no local intersection between two ellipsoids at the isolated tangent point.
Hence, the ellipsoids in (b) are in contact while the cylinders in (a) are not.}
\label{fig:cylinder_intersect}
\end{figure}

Let $X=(x,y,z,w)^T\in\mathbb{PR}^3$ and let two moving quadrics be given by 
$\mathcal{A}(t): X^{T}A(t)X=0$ and $\mathcal{B}(t): X^{T}B(t)X=0$, where
$A(t), B(t)$ are $4\times4$ matrices with elements as functions in
$t$.  The two quadrics define a moving pencil $\mathcal{Q}(\lambda;t): X^T \big( \lambda (A(t)-B(t)) \big) X = 0$,
with characteristic polynomial $f(\lambda;t)=\det\big(\lambda A(t)-B(t)\big)$.
We shall differentiate the cases in which the pencil $\mathcal{Q}(t)$ is (1) in general nondegenerate (i.e., $f(\lambda;t)\not\equiv 0$ for some $t$), or (2) always degenerate (i.e., $f(\lambda;t)\equiv 0$ for all $t$), 
and handle these two cases in different manners to be described in Section~\ref{sect:pencil_not_always_degenerate} and~\ref{sect:pencil_always_degenerate}, respectively.

%
%
%
%
%
%
%
%
%
%
%
%

\subsubsection{For $\mathcal{A}(t)$ and $\mathcal{B}(t)$ whose pencil is in general nondegenerate}
\label{sect:pencil_not_always_degenerate}

In this section, we consider two moving quadrics 
$\mathcal{A}(t)$ and $\mathcal{B}(t)$ whose pencil is in general nondegenerate, that is, 
their characteristic polynomial $f(\lambda)=\det\big(\lambda A(t)-B(t)\big)$
is not always identically zero over the time domain.

The morphologies of the intersection curves of two quadric surfaces (QSICs) in $\mathbb{PR}^3$ 
have been completely classified in~\cite{TuWMW2009}.
For two moving quadrics, the morphologies of their QSIC may change over time, and only some QSICs 
may correspond to a contact between the quadrics.
Therefore, our strategy is to first detect the time instants  (which we called the {\em candidate time instants}) at which 
two moving quadric surfaces have a change in their QSIC.
Our next step is then to identify whether a QSIC corresponds to a real contact (face, line or point contact) in $\mathbb{R}^3$ 
at each of the candidate time instants and to compute the contact between the two quadrics.

\paragraph{Determining candidate contact time instants}
The candidate time instants are the moments at which the QSIC of two moving quadrics change its morphological type.
To determine the candidate time instants, we make use of our recent result in detecting the variations of the QSIC of two moving quadrics in $\mathbb{PR}^3$.
Here, we give a brief idea of how this can be done and refer the reader to~\cite{JiaWC2012} for the details.
The classification by~\cite{TuWMW2009} distinguishes the QSIC types of two quadrics in $\mathbb{PR}^3$ from both algebraic and topological points of view (including singularities, number of components, and the degree of each irreducible component).
A QSIC type can be identified by the signature sequence and the Segre characteristics (\cite{Bromwich1906}) of the quadric pencil $Q(\lambda;t) = \lambda A(t) - B(t)$, which are in turn characterized by the algebraic properties of the roots of characteristic polynomial of $Q(\lambda;t)$, 
such as the number of real roots, the multiplicity of each root, and the type of the Jordan blocks associated with each root.
We proved that to detect all the time instants at which the QSIC changes is equivalent to detecting the time instants when the Segre characteristic of $Q(\lambda;t)$ changes.
This leads to an algebraic method using the techniques of resultants and Jordan forms to compute all the required time instants,
which, in our case, will serve as the candidate time instants for the next step.

\begin{remark}\label{rem:exact_solution}
With the aforementioned algebraic method, we obtain univariate equations defining the candidate time instants. The positive real solutions of such univariate equations can be computed efficiently by real root isolation solvers. Checking the sign of a polynomial expression at such a root can be done exactly by algebraic methods (see for example~\cite{BPR2006}).
\end{remark}

\paragraph{Identifying real contact}

For each of the candidate contact time instants $t_i$, the next step 
is to determine whether the QSIC corresponds to any real contact between the quadrics 
$\mathcal{A}(t_i)$ and $\mathcal{B}(t_i)$.
Table~\ref{tab:QSIC} is adapted from the three classification tables in~\cite{TuWMW2009} by showing only those cases\footnote{We keep the original case numbers for ease of reference.} in which the QSIC of two distinct quadrics corresponds to a point or a curve contact in $\mathbb{PR}^3$. 
It therefore encompasses all possible contact configurations of two quadrics.
The case of a face contact, that is, the two quadrics being identical, can be trivially identified and is therefore skipped here.
Now, for each candidate time instant $t_i$, we compute the signature sequence
of $\mathcal{A}(t_i)$ and $\mathcal{B}(t_i)$.
The quadrics have a contact in $\mathbb{PR}^3$ if and only if the sequence matches 
one of the 12 cases listed in Table~\ref{tab:QSIC}.
The following example shows how we may identify if two quadrics have a contact at a particular time instant by checking their Segre characteristics and signature sequence against Table~\ref{tab:QSIC}.

\begin{example}
Consider two cylinders, ${\mathcal A}: x^2 + z^2 = 1$ and ${\mathcal B}: y^2 + z^2 = 1$,
which have two singular intersection points as shown in Figure~\ref{fig:cylinder_intersect}(a).
The characteristic equation is $f(\lambda) = \lambda(\lambda-1)^2 = 0$.
The Segre characteristics is $[(11)11]_3$ and the signature sequence is $(2,((1,1)),2,(1,2),1,(1,2),2)$ which corresponds to case 13 of~\cite{TuWMW2009} in which the QSIC has two conics intersecting at two distinct non-isolated singular points.  
This case does not correspond to a contact configuration and is not listed in Table~\ref{tab:QSIC}.  The two cylinders are therefore not in contact.
\qed
\end{example}

\begin{table}[htbp]\caption{QSIC corresponding to a point contact or curve contact between two distinct quadrics in $\mathbb{PR}^3$.
In the illustrations, a solid line or curve represents a real component while a dashed one represents
an imaginary component. A solid dot indicates a real singular point.
A null-homotopic component is drawn as a closed loop,
and a non-null-homotopic component is shown as an open-ended curve.
A line or curve that is counted twice is thickened.
See~\cite{TuWMW2009} for details.
}\label{tab:QSIC}
\begin{center}
\begin{scriptsize}
\begin{tabular}{l|l|c|l|l}
\hline $[{\bf Segre}]_r$&\raisenumber{-2.8pt}{2.5 ex}{Case
\#}\hspace{-30pt}\rule[1.7 ex]{31pt}{0.3pt} \hspace{-5.3pt} \rule[1.7 ex]{0.3pt}{9.5pt}&&\\

$r$ = the \# &{\bf Index}&{\bf Signature Sequence}&{\bf Illus-}&{\bf Representative}\\
of real roots &{\bf Sequence}&&{\bf tration}& {\bf Quadric Pair}\\
\hline
%
%
\Lower{$[211]_3$}&\raisenumber{-2.8pt}{3.0 ex}{6 }\boundnum{2.2 ex}
{$\langle 1 {\wr\wr}_{-} 1 | 2 |3 \rangle$} &
(1,((1,2)),1,(1,2),2,(2,1),3) &\begin{minipage}{0.075\textwidth}
\centering
\includegraphics[width=\linewidth]{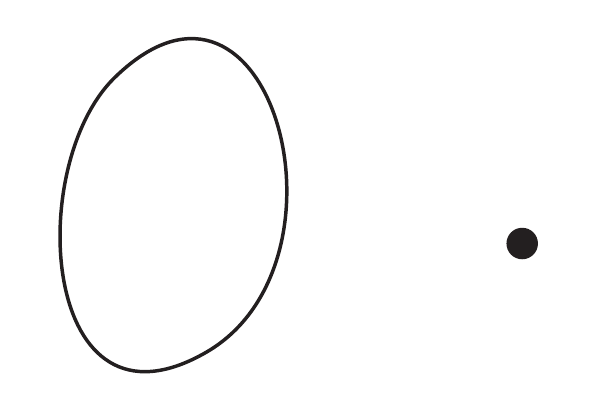}
\end{minipage}&
\begin{minipage}{0.18\textwidth}
\leftline{${\mathcal A}:\; -x^2 - z^2 +
2yw =0$} \leftline{${\mathcal B}:\; -3x^2 + y^2 - z^2
=0$}\end{minipage}
\\ \cline{2-5}
%
%
& \raisenumber{-2.8pt}{3.0 ex}{7 }\boundnum{2.2 ex} {$\langle 1
{\wr\wr}_{+} 1 | 2 |3 \rangle$} & (1,((0,3)),1,(1,2),2,(2,1),3)
&\begin{minipage}{0.075\textwidth}
\centering
\includegraphics[width=\linewidth]{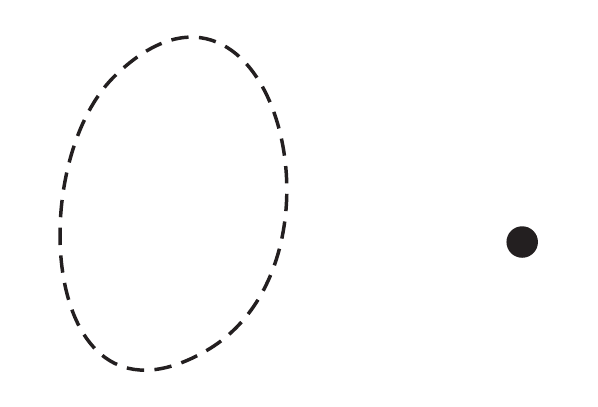}
\end{minipage}&
\begin{minipage}{0.18\textwidth}
\leftline{${\mathcal A}:\; x^2 + z^2 + 2yw
=0$} \leftline{${\mathcal B}:\; 3x^2 + y^2 + z^2 =0$}\end{minipage}
\\ \hline
%
%
$[(11)11]_3$ &\raisenumbertwo{-6.8pt}{3.0 ex}{15}\boundnum{2.2 ex} {$\langle 1 ||
1 | 2 | 3 \rangle$} &(1,((0,2)),1,(1,2),2,(2,1),3)
&\begin{minipage}{0.075\textwidth} 
\centering
\includegraphics[width=\linewidth]{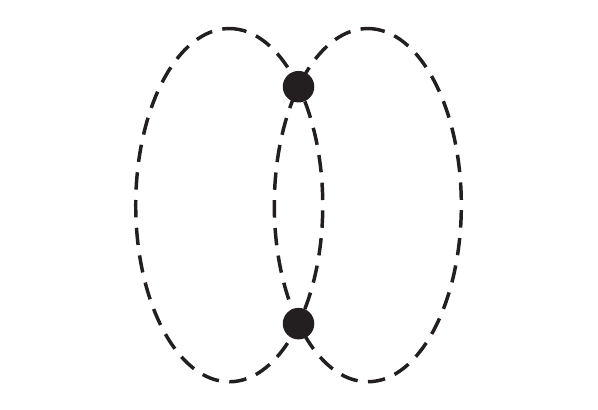}
\end{minipage}&
\begin{minipage}{0.26\textwidth}
\leftline{${\mathcal A}:\; x^2 + y^2 +z^2
- w^2=0$} \leftline{${\mathcal B}:\; x^2 + 2y^2=0$}\end{minipage}
\\ \hline
%
%
$[(111)1]_2$&\raisenumbertwo{-6.8pt}{3.0 ex}{19}\boundnum{2.2 ex} {$\langle 1 ||| 2 | 3 \rangle$} &
(1,(((0,1))),2,(2,1),3) &\begin{minipage}{0.075\textwidth}
\centering
\includegraphics[width=\linewidth]{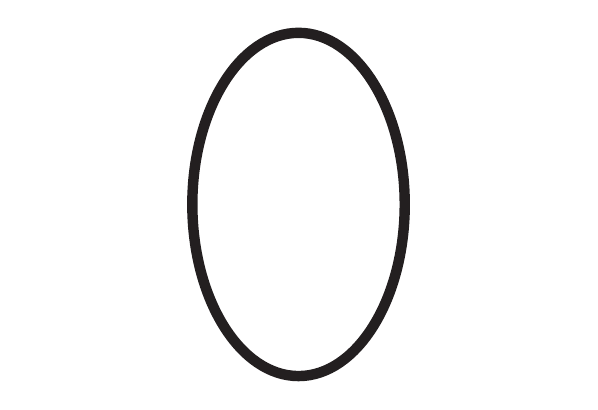}
\end{minipage}&
\begin{minipage}{0.26\textwidth}
\leftline{${\mathcal A}:\; y^2 + z^2 -w^2
=0$} \leftline{${\mathcal B}:\; x^2 =0$}\end{minipage}
\\  \hline
%
%
$[(21)1]_2$ &\raisenumbertwo{-6.8pt}{3.0 ex}{22}\boundnum{2.2 ex} {$\langle 1
{\wr\wr}_{+}| 2 | 3 \rangle$} & (1,(((0,2))),2,(2,1),3)
&\begin{minipage}{0.075\textwidth} 
\centering
\includegraphics[width=\linewidth]{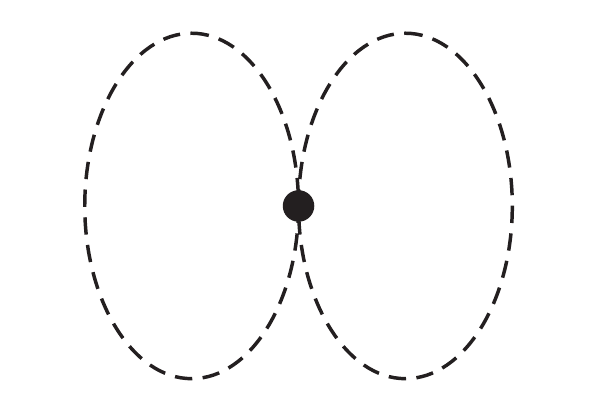}
\end{minipage}&
\begin{minipage}{0.26\textwidth}
\leftline{${\mathcal A}:\; y^2 - z^2 +2zw
=0$} \leftline{${\mathcal B}:\; x^2 + z^2 =0$}\end{minipage}
\\ \hline 
%
%
\Lower{$[2(11)]_2$}&\raisenumbertwo{-6.8pt}{3.0 ex}{24}\boundnum{2.2 ex}
{$\langle 1 {\wr\wr}_{-} 1 || 3 \rangle$} &
(1,((1,2)),1,((1,1)),3)&\begin{minipage}{0.075\textwidth}
\centering
\includegraphics[width=\linewidth]{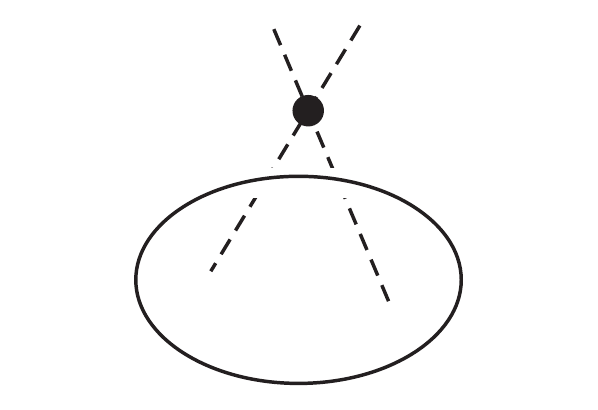}
\end{minipage}&
\begin{minipage}{0.26\textwidth}
\leftline{${\mathcal A}:\; 2xy - y^2=0$}
\leftline{${\mathcal B}:\; y^2  - z^2 - w^2=0$}\end{minipage}
\\  \cline{2-5}
%
%
&\raisenumbertwo{-6.8pt}{3.0 ex}{25}\boundnum{2.2 ex} {$\langle 1
{\wr\wr}_{+} 1 || 3 \rangle$}& (1,((0,3)),1,((1,1)),3)
&\begin{minipage}{0.075\textwidth} 
\centering
\includegraphics[width=\linewidth]{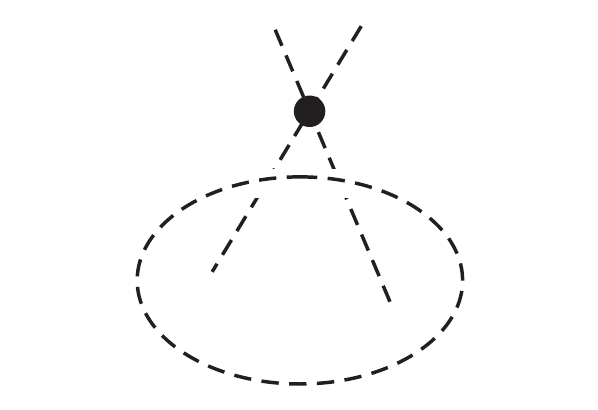}
\end{minipage}&
\begin{minipage}{0.26\textwidth}
\leftline{${\mathcal A}:\; 2xy - y^2=0$}
\leftline{${\mathcal B}:\; y^2 + z^2 + w^2=0$}\end{minipage}
\\
\hline
%
%
$[(11)(11)]_2$ &\raisenumbertwo{-6.8pt}{3.0 ex}{30}\boundnum{2.2 ex} {$\langle 1 || 1 || 3
\rangle$} &(1,((0,2)),1,((1,1)),3)
&\begin{minipage}{0.075\textwidth} 
\centering
\includegraphics[width=\linewidth]{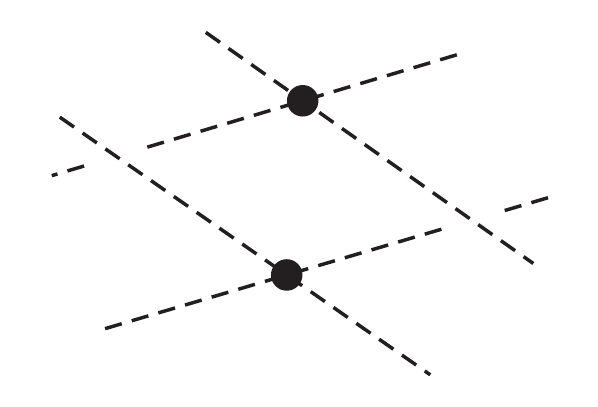}
\end{minipage}&
\begin{minipage}{0.26\textwidth}
\leftline{${\mathcal A}:\; x^2 + y^2 =0$}
\leftline{${\mathcal B}:\; z^2-w^2=0$}\end{minipage}
\\  \hline
%
%
%
\Lower{$[(211)]_1$}&\raisenumbertwo{-6.8pt}{3.0
ex}{32}\boundnum{2.2 ex} {$\langle 2 {\wr\wr}_{-}|| 2 \rangle$}  &
(2,((((1,0)))),2)&\begin{minipage}{0.075\textwidth}
\centering
\includegraphics[width=\linewidth]{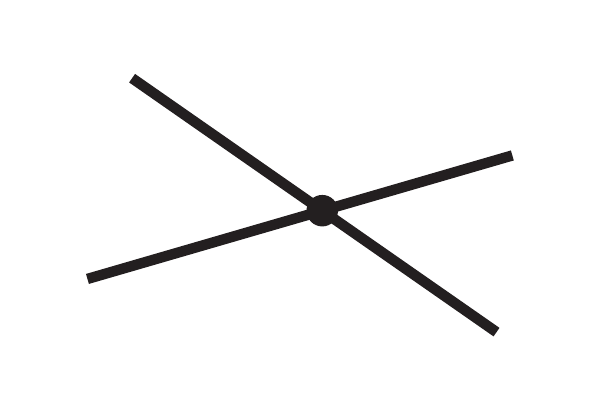}
\end{minipage}&
\begin{minipage}{0.26\textwidth}
\leftline{${\mathcal A}:\; x^2-y^2+2zw
=0$} \leftline{${\mathcal B}:\; z^2=0$}\end{minipage}
\\ \cline{2-5}
%
%
&\raisenumbertwo{-6.8pt}{3.0 ex}{33}\boundnum{2.2 ex} {$\langle 1
{\wr\wr}_{-}|| 3 \rangle$}  & (1,((((1,0)))),3)
&\begin{minipage}{0.075\textwidth} 
\centering
\includegraphics[width=\linewidth]{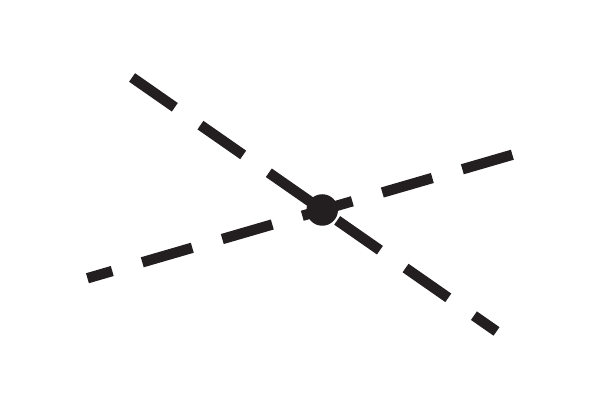}
\end{minipage}&
\begin{minipage}{0.26\textwidth}
\leftline{${\mathcal A}:\; x^2+y^2+2zw
=0$} \leftline{${\mathcal B}:\; z^2=0$}\end{minipage}
\\ \hline 
%
%
%
\Lower{$[(22)]_1$}&\raisenumbertwo{-6.8pt}{3.0 ex}{34}\boundnum{2.2
ex} {$\langle 2 \widehat{\wr\wr}_{-}\widehat{\wr\wr}_{-} 2 \rangle$}
& (2,((((2,0)))),2) &\begin{minipage}{0.075\textwidth}
\centering
\includegraphics[width=\linewidth]{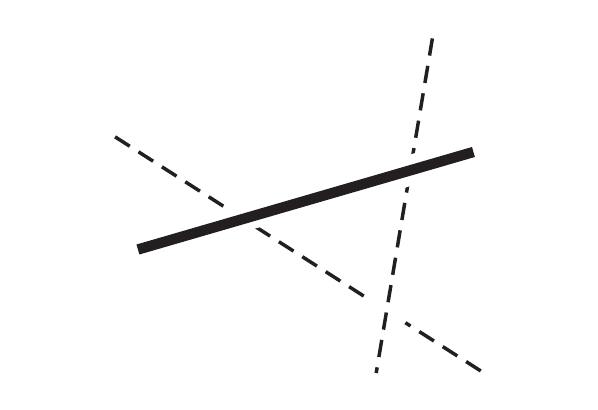}
\end{minipage}&
\begin{minipage}{0.26\textwidth}
\leftline{${\mathcal A}:\; xy + zw =0$}
\leftline{${\mathcal B}:\; y^2 + w^2=0$}\end{minipage}
\\ \cline{2-5}
%
%
&\raisenumbertwo{-6.8pt}{3.0 ex}{35}\boundnum{2.2 ex} {$\langle 2
\widehat{\wr\wr}_{-}\widehat{\wr\wr}_{+} 2 \rangle$} &
(2,((((1,1)))),2) &\begin{minipage}{0.075\textwidth}
\centering
\includegraphics[width=\linewidth]{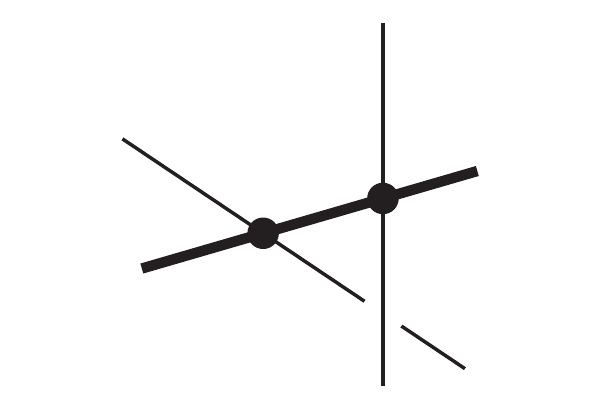}
\end{minipage}&
\begin{minipage}{0.26\textwidth}
\leftline{${\mathcal A}:\; xy - zw =0$}
\leftline{${\mathcal B}:\; y^2 - w^2=0$}\end{minipage}
\\ \hline 
\end{tabular}
\end{scriptsize}
\end{center}
\end{table}

\paragraph{Computing contact}
We can proceed to compute a contact once it is identified.
The following lemma provides a means to computing the contact points of two quadrics at a particular time instant:

\begin{lemma}\label{lem:cqm_touch_config}
Let ${\mathcal A}: X^TAX=0$ and ${\mathcal B}: X^TBX=0$ be two distinct, irreducible quadric surfaces whose pencil is nondegenerate.
Suppose that ${\mathcal A}$ and ${\mathcal B}$ are in contact (i.e., whose QSIC is listed in Table~\ref{tab:QSIC}), and 
let $\lambda_0$ be a multiple root of $f(\lambda)=\det(\lambda A-B) = 0$,
the characteristic equation of ${\mathcal A}$ and ${\mathcal B}$.
Then, we have the following cases:
\begin{enumerate}
\item If $rank(\lambda_0 A - B) = 3$, $\lambda_0$ corresponds to
      one singular intersection point $\mathbf p$ of $\mathcal A$ and $\mathcal B$ in $\mathbb{PR}^3$.  If $\lambda_0 \not=0$, $\mathcal A$ and $\mathcal B$ are tangential at $\mathbf p$; otherwise, $\mathcal B$ is a cone with $\mathbf p$ as its apex which lies also on $\mathcal A$.
\item If $rank(\lambda_0 A - B) = 2$, $\lambda_0$ corresponds to singular intersection between $\mathcal A$ and $\mathcal B$ that happens at 
	  either one point, two distinct points, or 
	  along a straight line in $\mathbb{PC}^3$, where $\mathcal A$ and $\mathcal B$ are tangential to each other.
\item If $rank(\lambda_0 A - B) = 1$, $\lambda_0$ corresponds to
      singular intersection between $\mathcal A$ and $\mathcal B$
      along a conic curve 
      (which can be a reducible one) 
      in $\mathbb{PC}^3$, where $\mathcal A$ and $\mathcal B$ are tangential to each other.
\end{enumerate}
\end{lemma}

The proof is given in Appendix~\ref{app:proof}.

Note that when $f(\lambda)=0$ has more than one multiple root,
we should consider all its multiple roots in order to obtain all contact points between the two quadrics.
According to Lemma~\ref{lem:cqm_touch_config}, given two touching quadric surfaces at time $t_i$, the contact points are in general the solutions of $\big(\lambda_j A(t_i) - B(t_i)\big)X=0$ for each multiple root $\lambda_j$ of $f(\lambda;t_i)=0$.  
We also need to differentiate between real and imaginary contacts.
For example 
in both cases 24 and 25 of Table~\ref{tab:QSIC}, 
the characteristic equation has two multiple roots $\lambda_0$ and $\lambda_1$, with
$rank(\lambda_0 A - B)=3$ and $rank(\lambda_1 A - B)=2$;  
$\lambda_0$ corresponds to a real contact point while $\lambda_1$ corresponds to two distinct imaginary contacts
which should be discarded.

\begin{example}
Consider the unit sphere $\mathcal{A}: x^2 + y^2 + z^2 = 1$ and a cylinder 
$\mathcal{B}: x^2 + y^2 = 1$.  
The characteristic equation of $\mathcal{A}$ and $\mathcal{B}$
is $f(\lambda)= -\lambda(\lambda-1)^3$ which has a triple root $\lambda_0=1$.
Also, $rank(\lambda_0 A - B) = 1$ and by Lemma~\ref{lem:cqm_touch_config},
$\lambda_0$ corresponds to a contact along a conic curve between $\mathcal{A}$
and $\mathcal{B}$.  
Now, $(\lambda_0 A -B )X = 0$ has three linearly independent solutions
$X_0 = (0,0,0,1)^T$, $X_1 = (1,0,0,0)^T$ and $X_2 = (0,1,0,0)^T$ which span 
the plane $z=0$.
Intersecting the plane $z=0$ with $\mathcal{A}$ yields the circle $x^2+y^2=1, z=0$,
which is the contact between $\mathcal{A}$ and $\mathcal{B}$.
\qed
\end{example}

Cases 6, 24 and 35 are situations in which the quadrics are tangent at some regions but at the same time having local real intersection at the others.
Since we assume that the CQMs are separate initially and we seek their first contact, it can be assured that a local real intersection must take place after a proper contact is found.
The real intersections in these cases can therefore be ignored.

The contact points computed so far are between the quadric surfaces but not necessarily between the CQMs, 
therefore all contact points are further subject to validation to see if they are on both CQMs.
Contact points at infinity are thus discarded.
Validation details will be discussed in Section~\ref{sect:contact_validation}.

\subsubsection{For $\mathcal{A}(t)$ and $\mathcal{B}(t)$ whose pencil is always degenerate}
\label{sect:pencil_always_degenerate}

We now consider the case of two moving quadrics which always define a degenerate pencil, that is,
$f(\lambda;t)\equiv 0$ for all $t$.
Here, all members of the pencils are projective cones for all 
$t$ (\cite{Farouki1989}), which means that 
the vertices of the projective cones always lie on a common generator of the cones and the cones are always tangential along this generator in $\mathbb{PC}^3$.
Considering any affine realization of the projective space,
Figure~\ref{fig:projective_cones} depicts the three situations of this kind that are only possible in $\mathbb{R}^3$:
(a) Two moving cones whose apexes slide along the common generator, 
(b) two moving cylinders whose axes are always parallel (with their ``apexes'' at infinity),
and (c) a moving cone and a moving cylinder which always share a common generator.
Note that a cylinder can either be an elliptic, a hyperbolic or a parabolic cylinder.
For case (a), we need only to consider the contact of the vertices, since the cones are initially separate and any other contact configurations can be detected by CCD of other boundary elements of the CQMs.  Hence, 
the candidate contact time instants are the contact time instants of the vertices of the cones.
For (b),  the CCD problem is transformed to a two-dimensional CCD of two moving conics on a plane $\mathcal{P}$ 
orthogonal to the cylinder axes, with the conics being the cross-sections of the cylinders on $\mathcal{P}$.
CCD of moving conics in $\mathbb{R}^2$ will be discussed in Section~\ref{sect:CCD_conics_2D}.
We may disregard case (c) for a moving cone and a moving cylinder, since a cylinder must be delimited by a boundary curve on a CQM and any possible first contact can be captured by CCD of other boundary elements of the CQMs.

\begin{figure}[!ht]
\centering
\begin{minipage}{\linewidth}
\centering
\begin{minipage}{0.3\linewidth}
\centering
\includegraphics[width=\linewidth]{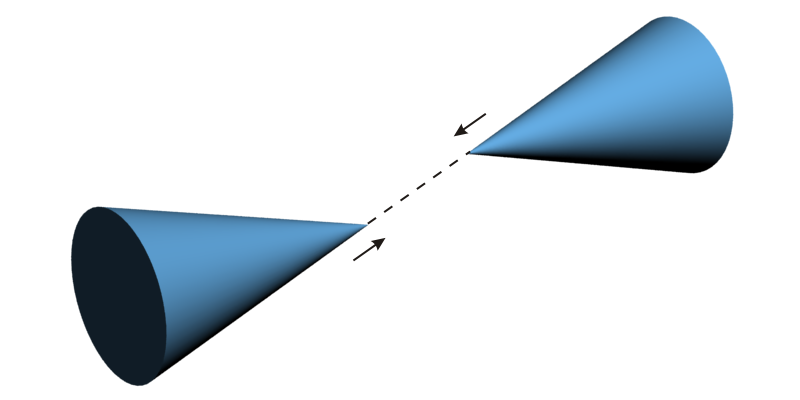}\\
{\small (a)}
\end{minipage}
\hfill
\begin{minipage}{0.3\linewidth}
\centering
\includegraphics[width=\linewidth]{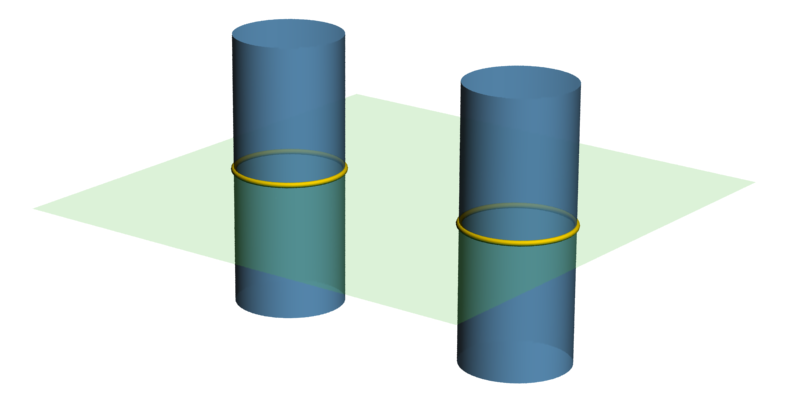}\\
{\small (b)}
\end{minipage}
\hfill
\begin{minipage}{0.3\linewidth}
\centering
\includegraphics[width=\linewidth]{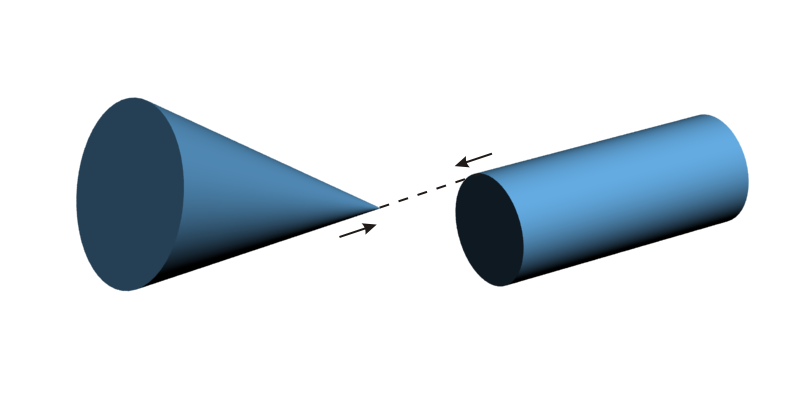}\\
{\small (c)}
\end{minipage}
\end{minipage}
\caption{The three possible scenarios of two moving quadrics in $\mathbb{R}^3$ whose pencil is always degenerate.
(a) Two moving cones whose vertices always lie on a common generator; (b) two moving cylinders whose axes are always parallel; and (c) a moving cone and a moving cylinder which always share a common generator.  Note that a cylinder may either be an elliptic, a hyperbolic or a parabolic cylinder. Also, only one nappe of a cone is shown.}
\label{fig:projective_cones}
\end{figure}

\subsection{Case~\casequadricsplanes --- quadrics vs. planes}
\label{sect:ccd_quadrics_planes}

We first note that the singular case where a plane is in contact with a cone only at its apex
is not considered here, as the contact can be determined directly by an $(F,V)$-type CCD of the apex and the plane.

Now, consider an irreducible quadric surface ${\mathcal A}(t): X^T A(t) X = 0$ and 
a plane ${\mathcal N}(t): N(t)^T X = 0$ in ${\mathbb R}^3$.
A necessary and sufficient condition for ${\mathcal N}(t)$ to be a tangent plane 
to ${\mathcal A}(t)$ at some point $X_0 \in {\mathbb R}^3$ 
is that 
\[
\begin{cases}
\;\alpha N(t) = A(t) X_0 & \text{for some nonzero $\alpha \in \mathbb R$, and} \\
\; N(t)^T X_0 = 0. & 
\end{cases}
\]
These two equations can be written as 
\begin{equation*}
\begin{pmatrix} A(t) & N(t) \\  N(t)^T & 0 \end{pmatrix}
\begin{pmatrix} X_0 \\ -\alpha \end{pmatrix} = 0,
\end{equation*}
which has a nonzero solution
$(X_0 \;\; \alpha)^T$ if
\begin{equation}\label{eqn:quadric_plane}
	\begin{vmatrix} A(t) & N(t) \\  N(t)^T & 0 \end{vmatrix} = 0.
\end{equation}
Therefore, the roots of Eq.~(\ref{eqn:quadric_plane}) corresponding to 
a solution $(X_0 \;\; \alpha)^T$ with $\alpha \not=0$ 
yield the candidate time instants 
of the quadric ${\mathcal A}(t)$ and the plane ${\mathcal N}(t)$.

\subsection{Case~\casequadricsconics --- quadrics vs. conics}
\label{sect:CCD_quadrics_conics}

We adopt a {\em dimension reduction technique}
to reduce CCD of extended
boundary elements to CCD of primitives of lower dimensions.
By doing so, we also simplify the algebraic formulations.
Figure~\ref{capped cylinders} illustrates CCD of
two moving capped elliptic cylinders. In this example, there are 
three cases to which the dimension reduction technique
can be applied, namely, quadrics vs.\ conics (Figure~\ref{capped cylinders}(b)),
planes vs.\ conics (Figure~\ref{capped cylinders}(c))
and conics vs.\ conics (Figure~\ref{capped cylinders}(d)).
The reduction for the latter two cases will be discussed in subsequent sections.

\begin{figure}[!ht]
\centering
  \begin{minipage}{\linewidth}
  \begin{minipage}{.24\linewidth}\centering
    \includegraphics[width=\linewidth]{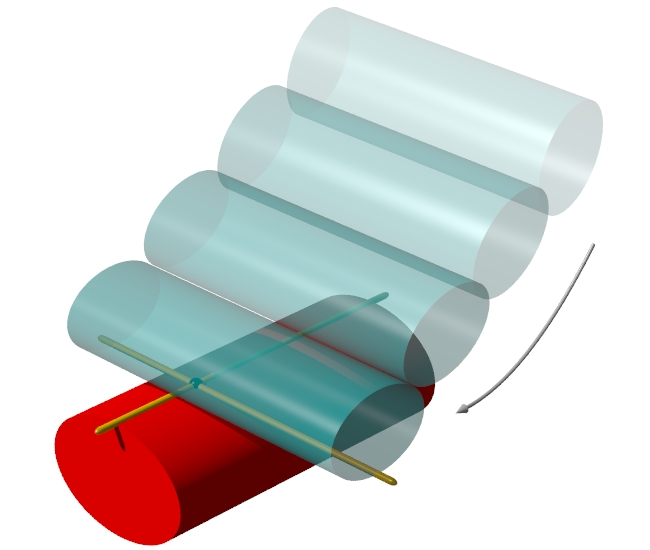}\\ 
    {\small (a)}
  \end{minipage}\hfill
  \begin{minipage}{.24\linewidth}\centering
    \includegraphics[width=\linewidth]{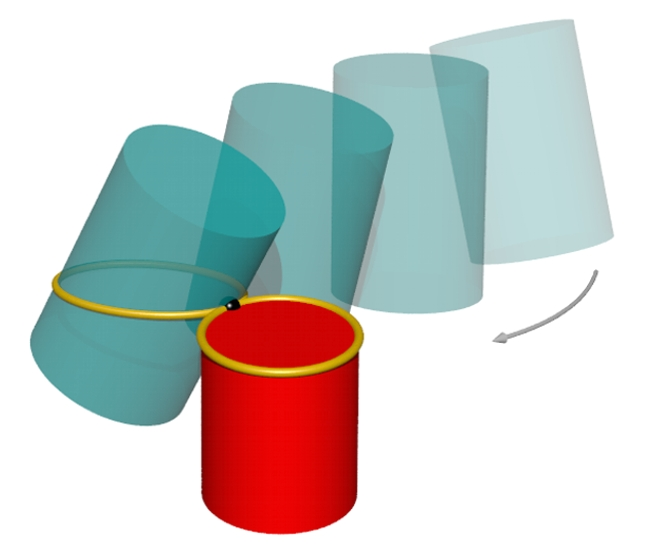}\\ 
    {\small (b)}
  \end{minipage}\hfill 
   \begin{minipage}{.24\linewidth}\centering
    \includegraphics[width=\linewidth]{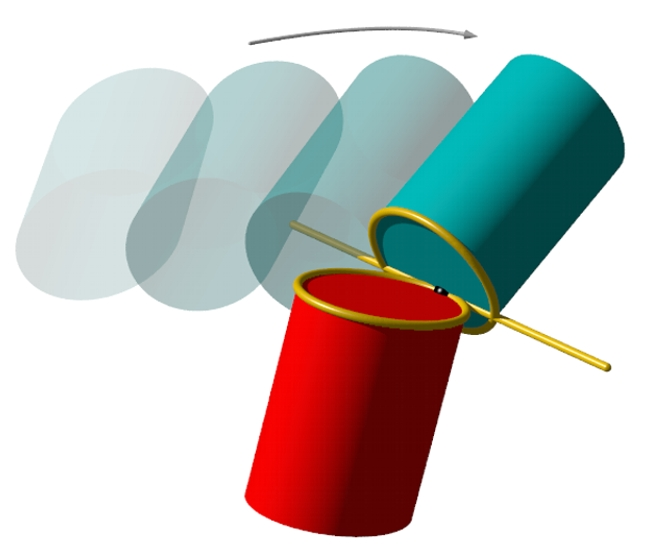}\\ 
    {\small (c)}
  \end{minipage}\hfill
  \begin{minipage}{.24\linewidth}\centering
    \includegraphics[width=\linewidth]{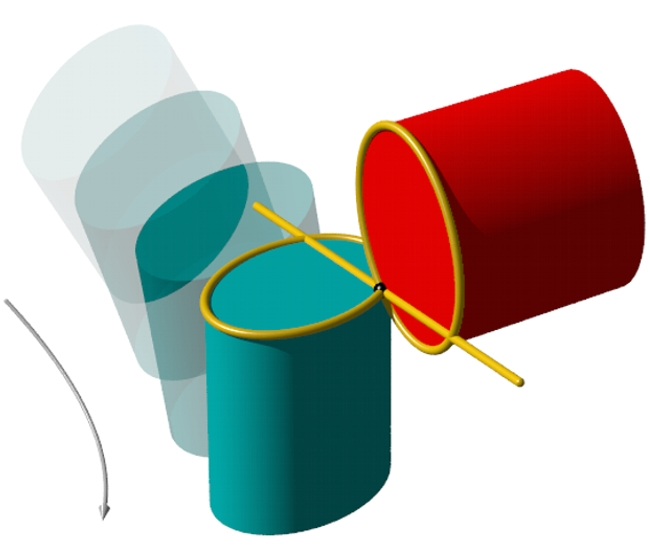}\\ 
    {\small (d)}
  \end{minipage}
  \end{minipage}
\caption{Contact configurations of two capped cylinders determined by CCDs of four different types of element pair.
(a) $(F,F)$-type; (b \& c) $(F,E)$-type; and (d) $(E,E)$-type.  CCDs of (b)--(d) are solved using the dimtension
reduction technique.}
\label{capped cylinders}
\end{figure}

We first consider CCD of a quadric surface ${\mathcal S}(t): X^T S(t) X = 0$
and a conic curve ${\mathcal C}(t)$ defined in the plane 
$\Pi_{\mathcal C}(t)$ in $\mathbb R^3$.
Let ${\mathcal S}_\Pi(t)$ be the intersection of ${\mathcal S}(t)$ with $\Pi_{\mathcal C}(t)$.
We thereby reduce CCD of ${\mathcal S}(t)$ and ${\mathcal C}(t)$ to CCD
of ${\mathcal S}_\Pi(t)$ and ${\mathcal C}(t)$, which are two conics, in the plane $\Pi_{\mathcal C}(t)$.
CCD of two conics in $\mathbb R^2$
is handled in Case~\caseconicsconicstwo (Section~\ref{sect:CCD_conics_2D}).

\subsection{Case~\caseplanesconics --- planes vs. conics}
\label{sect:CCD_planes_conics}

Consider CCD of a plane ${\mathcal P}(t)$
and a conic curve ${\mathcal C}(t)$ defined in a plane $\Pi_{\mathcal C}(t)$ in $\mathbb R^3$.
The case of ${\mathcal P}(t)$ and $\Pi_{\mathcal C}(t)$ being identical for all $t$
can be disregarded, as any possible first contact of the CQMs due to 
${\mathcal P}(t)$ and ${\mathcal C}(t)$
can then be detected by CCD of ${\mathcal C}(t)$ and other boundary elements on 
${\mathcal P}(t)$. 
If ${\mathcal P}(t)$ and $\Pi_{\mathcal C}(t)$ are parallel for all $t$,
then there is no contact between ${\mathcal P}(t)$ and ${\mathcal C}(t)$.
Otherwise, with dimension reduction, CCD of a plane ${\mathcal P}(t)$
and a conic curve ${\mathcal C}(t)$ is reduced to 
CCD of ${\mathcal C}(t)$ and a moving line which is the 
intersection between ${\mathcal P}(t)$ and $\Pi_{\mathcal C}(t)$ 
in $\mathbb{R}^2$, and the latter is handled by Case~\caseconicslinestwo (Section~\ref{sect:CCD_conics_lines_2D})
for CCD between conics and lines in $\mathbb{R}^2$.
For the candidate time instants thus found, we will verify and discard those $t_i$ at which ${\mathcal P}(t_i)$ and $\Pi_{\mathcal C}(t_i)$ are parallel.

\subsection{Case~\casequadricslines --- quadrics vs. lines}
\label{sect:CCD_quadrics_lines}

Suppose ${\mathcal S}(t): X^T S(t) X = 0$ is a quadric surface
and ${\mathcal L}(u;t)$ is a line in $\mathbb R^3$.
We simply substitute ${\mathcal L}(u;t)$
into ${\mathcal S}(t)$ and obtain $g(u;t) = L(u;t)^T S(t) L(u;t)$
which is quadratic in $u$.  
The line ${\mathcal L}(u;t_i)$ touches ${\mathcal S}(t_i)$ at a particular time $t_i$
if $g(u;t_i)$ has a double root $u_0$.
Hence, the candidate contact time instants
of ${\mathcal S}(t)$ and ${\mathcal L}(u;t)$
are given by the roots of the discriminant $\Delta_g(t)$ of $g(u;t)$.
For each candidate contact time instant $t_i$, the contact point is  
${\mathcal L}(u_0;t_i)$ where $u_0$ is the double root of $g(u;t_i)$.
If $g(u;t_i)$ is identically zero, we have ${\mathcal L}(u;t_i)$ lying entirely on 
${\mathcal S}(t_i)$.

%

\subsection{Case~\casequadricsvertices --- quadrics/planes vs. vertices}
\label{sect:CCD_quadrics_vertices}

Let ${\mathcal S}(t): X^T S(t) X = 0$ be a quadric surface and
${\bf p}(t)$ be a vertex in $\mathbb R^3$.
By direct substitution, we obtain the equation
${\bf p}^T(t) S(t) {\bf p}(t) = 0$
whose roots give the candidate contact time instants.
Similarly, for CCD of a plane ${\mathcal H}(t): H(t)^T X = 0$ and
a vertex ${\bf p}(t)$,
the candidate contact time instants are the roots of the equation
$H(t)^T {\bf p}(t) = 0$.

\subsection{Case~\caseconicsconicsthree --- conics vs. conics in $\mathbb R^3$}
\label{sect:CCD_space_conics}

We transform the problem of CCD of two conics in $\mathbb R^3$ into 
CCD of one dimension in the line where the containing planes of the conics intersect.
Let a moving conic ${\mathcal A}(t)$ be defined as the intersection between a quadric $\tilde{\mathcal A}(t): X^T A(t)X=0$ and a plane $\Pi_A(t)$ in $\mathbb{R}^3$.
Similarly, ${\mathcal B}(t)$ is a moving conic which is the intersection between a quadric $\tilde{\mathcal B}(t): X^T B(t)X=0$ and a plane $\Pi_B(t)$ in $\mathbb{R}^3$.
We first assume that $\Pi_A(t) \not\equiv \Pi_B(t)$ and
let ${\mathcal L}(u;t)$ be a parameterization of the line of intersection
between $\Pi_A(t)$ and $\Pi_B(t)$.

Substituting ${\mathcal L}(u;t)$ into the conic equations,
we have:
\begin{eqnarray}
h(u;t): \;\; L^T(u;t)A(t)L(u;t) = 0,\label{eqn:subsL}\\
g(u;t): \;\; L^T(u;t)B(t)L(u;t) = 0.\nonumber
\end{eqnarray}

The solution of $h(u;t) = 0$ gives the intersection between ${\mathcal L}(u;t)$ and
${\mathcal A}(t)$.
Likewise, the solution of $g(u;t) = 0$ gives the intersection between ${\mathcal L}(u;t)$
and ${\mathcal B}(t)$.
Since $h(u;t)$ and $g(u;t)$ are quadratic in $u$, we may write
$h(u;t) = U^T {\hat A}(t) U$ and
$g(u;t) = U^T {\hat B}(t) U$
where $U = (u,1)^T$ and ${\hat A}(t)$ and ${\hat B}(t)$ are $2 \times 2$ coefficient matrices.
It means that $h(u;t)$ and $g(u;t)$ can be considered as two moving 
``1D projective conics''
(i.e., intervals or line segments),
denoted by $\hat{\mathcal A}(t)$ and $\hat{\mathcal B}(t)$,
which can be either real or imaginary.
Now, ${\mathcal A}(t)$ and ${\mathcal B}(t)$ have real tangency in 
$\mathbb{PR}^3$
if and only if there is real tangency between $\hat{\mathcal A}(t)$ and $\hat{\mathcal B}(t)$ in $\mathbb{PR}^3$, that is, an end-point of $\hat{\mathcal A}(t)$ overlap with an end-point of $\hat{\mathcal B}(t)$.
Hence, we have essentially reduced a 3D problem (namely, CCD of two moving conics in the space) to a 1D problem (namely, CCD of two moving intervals in a line).

Let $f(\lambda) = \det (\lambda {\hat A} - {\hat B})$ be the characteristic polynomial of
two static 1D conics $\hat{\mathcal A}: X^T{\hat A}X=0$ and
$\hat{\mathcal B}: X^T{\hat B}X =0$ in $\mathbb{PR}$.
The intersection of $\hat{\mathcal A}$ and $\hat{\mathcal B}$ can be
characterized by the roots of $f(\lambda)$,
as summarized in~Table~\ref{tab:1D_ellipse}
(the derivation follows similarly as in~\cite{Choi2008} for the characterization 
of the intersection of two 1D ellipses in $\mathbb{R}$).
Hence, we have the following theorem stating the conditions for 
two conics to have contact in $\mathbb{PR}^3$:

\begin{table}[ht!]
\caption{Configuration of two 1D conics $\hat{\mathcal A}$ \&
$\hat{\mathcal B}$ in $\mathbb{PR}$ and the roots of their characteristic equation $f(\lambda)=0$.
Each 1D conic is represented in pairs of brackets of the same style.
Degenerate conic of one point is represented by either a dot or a cross.
}\label{tab:1D_ellipse}
\begin{center}
\begin{tabular}{ccl}\hline
\rule[-1.5ex]{0ex}{4ex}Roots of $f(\lambda)=0$ &  Configuration & \\
\hline
\parbox{.3\linewidth}{(1) Distinct positive} &
\parbox{.2\linewidth}{\includegraphics[width=\linewidth]{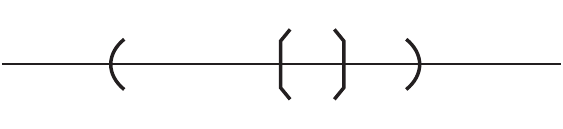}} &
or both $\hat{\mathcal A}$ \& $\hat{\mathcal B}$ are imaginary and 
$\hat{\mathcal A} \not= \hat{\mathcal B}$\\
\parbox{.3\linewidth}{(2) Distinct negative} &
\parbox{.2\linewidth}{\includegraphics[width=\linewidth]{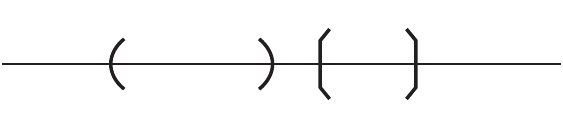}} & \\
\parbox{.3\linewidth}{(3) One zero, one positive} &
\parbox{.2\linewidth}{\includegraphics[width=\linewidth]{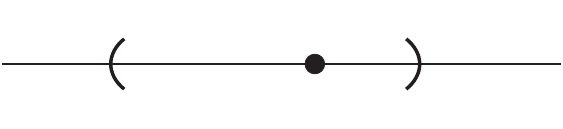}} & \\
\parbox{.3\linewidth}{(4) One zero, one negative} &
\parbox{.2\linewidth}{\includegraphics[width=\linewidth]{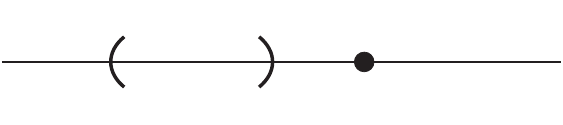}} & \\
\parbox{.3\linewidth}{(5) One negative, one positive} &
\parbox{.2\linewidth}{\includegraphics[width=\linewidth]{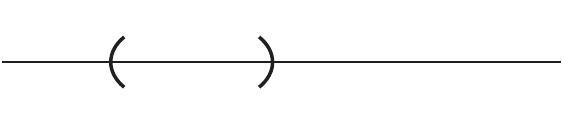}} &
with another conic being imaginary\\
\parbox{.3\linewidth}{(6) Positive double} &
\parbox{.2\linewidth}{\includegraphics[width=\linewidth]{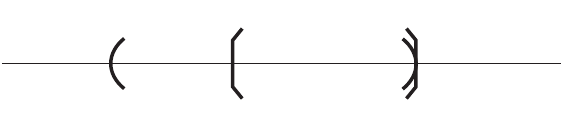}} &
or both $\hat{\mathcal A}$ \& $\hat{\mathcal B}$ are imaginary and 
$\hat{\mathcal A} = \hat{\mathcal B}$\\
\parbox{.3\linewidth}{(7) Negative double} &
\parbox{.2\linewidth}{\includegraphics[width=\linewidth]{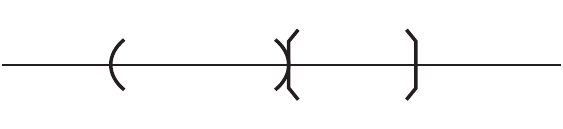}} & \\
\parbox{.3\linewidth}{(8) Double zero} &
\parbox{.2\linewidth}{\includegraphics[width=\linewidth]{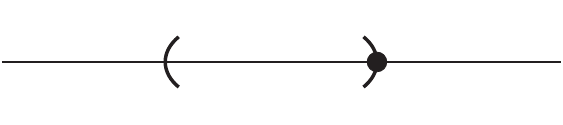}} &\\
\parbox{.3\linewidth}{(9) Complex conjugate} &
\parbox{.2\linewidth}{\includegraphics[width=\linewidth]{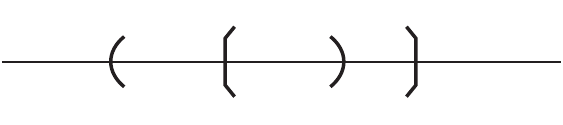}} & \\
\parbox{.3\linewidth}{(10) $f(\lambda)$ is linear with roots $=0$}  &
\parbox{.2\linewidth}{\includegraphics[width=\linewidth]{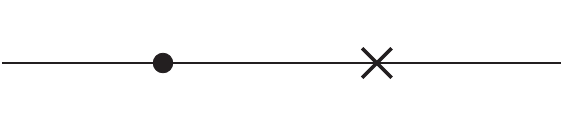}} & \\
\parbox{.3\linewidth}{(11) $f(\lambda) \equiv 0$} &
\parbox{.2\linewidth}{\includegraphics[width=\linewidth]{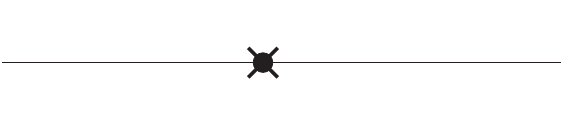}} & \\\\
\hline
\end{tabular}
\end{center}
\end{table}

\begin{theorem}\label{thm:conics3D}
Given two conics $\mathcal A$ (on plane $\Pi_A$)
and $\mathcal B$ (on plane $\Pi_B$) in ${\mathbb R}^3$, suppose that
$\Pi_A$ and $\Pi_B$ intersect at some line $\mathcal L \in {\mathbb R}^3$.
Let $\hat{\mathcal A}: X^T{\hat A}X=0$ and $\hat{\mathcal B}: X^T{\hat B}X = 0$ be the ``1D conics''
characterizing the intersections of $\mathcal L$ with $\mathcal A$ and
$\mathcal B$, respectively.
Furthermore, let $f(\lambda)=\det(\lambda {\hat A}-{\hat B})$ be the characteristic
polynomial of $\hat{\mathcal A}$ and $\hat{\mathcal B}$.  Then, the conics
${\mathcal A}$ and ${\mathcal B}$ are in contact in $\mathbb{PR}^3$
if and only if
\begin{enumerate}
\item $f(\lambda)$ has a double root (Figure~\ref{fig:ellipse3Dthm}(a-c)); or
\item $f(\lambda) \equiv 0$ (Figure~\ref{fig:ellipse3Dthm}(d)).
\end{enumerate}
\end{theorem}

\begin{figure*}[ht!]
\centering
  \begin{minipage}{\linewidth}
  \begin{minipage}{.24\linewidth}\centering
    \includegraphics[width=\linewidth]{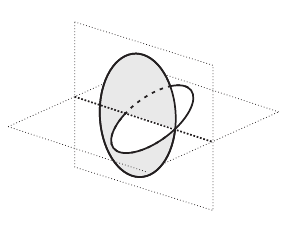}\\ \vspace{-1ex}
    {\scriptsize (a)}
  \end{minipage}\hfill
  \begin{minipage}{.24\linewidth}\centering
    \includegraphics[width=\linewidth]{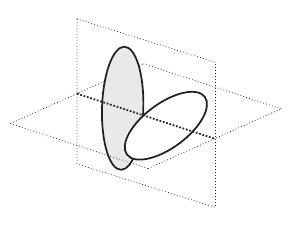}\\ \vspace{-1ex}
    {\scriptsize (b)}
  \end{minipage}\hfill
  \begin{minipage}{.24\linewidth}\centering
    \includegraphics[width=\linewidth]{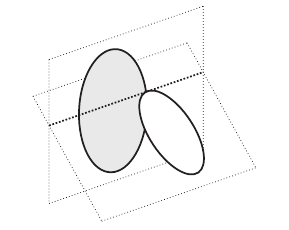}\\ \vspace{-1ex}
    {\scriptsize (c)}
  \end{minipage}\hfill
  \begin{minipage}{.24\linewidth}\centering
    \includegraphics[width=\linewidth]{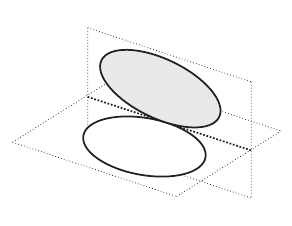}\\ \vspace{-1ex}
    {\scriptsize (d)}
  \end{minipage}
  \end{minipage}
\caption{The four configurations of two touching conics in 3D.
Sub-figures (a), (b), (c) \& (d) correspond to the cases (6), (7), (8) \& (11)
of Table~\ref{tab:1D_ellipse}, respectively.}
\label{fig:ellipse3Dthm}
\end{figure*}

Algorithm~\ref{alg:CCD_conics} gives the procedure for solving CCD of two moving conics in $\mathbb R^3$.  First of all,
if two moving conics are found to be contained in the same plane (i.e., $\Pi_A(t) \equiv \Pi_B(t)$) for all $t$, we may apply a continuous transformation $M(t)$ to both conics that maps $\Pi_A(t)$ to $\mathbb R^2$ and the problem is reduced to a two dimensional CCD of two moving conics in $\mathbb R^2$.
Otherwise, we reduce the problem to a one dimensional CCD of two 1D conics using the above formulation and capture the time instants at which the conditions in Theorem~\ref{thm:conics3D} are satisfied.
This is done by computing the zeroes of the discriminant $\Delta_f(t)$ of $f(\lambda;t)$ which give the instants $t_i$
when $f(\lambda;t_i)$ has a double root (condition 1) or $f(\lambda;t_i) \equiv 0$
(condition 2).
\label{res:plane_parallel}
The \textbf{for} loop in the algorithm handles the special case in which a zero $t_i$ of $\Delta_f(t)$ corresponds to when $f(\lambda;t_i) \equiv 0$.
This may happen when the containing planes of both conics $\mathcal{A}(t_i)$ and 
$\mathcal{B}(t_i)$ are parallel so that ${\cal L}(u;t_i)$ is a line at infinity, and the conics are not in contact.
The function $f(\lambda;t_i)$ may also be identically zero when
$\mathcal{A}(t_i)$ and 
$\mathcal{B}(t_i)$ lie on the same plane so that ${\cal L}(u;t_i)$ becomes undefined.
The two conics may or may not have a contact in this case and therefore we need to further carry out a 2D static collision detection of ${\cal A}'(t_i)$ and ${\cal B}'(t_i)$, the image of 
${\cal A}(t_i)$ and ${\cal B}(t_i)$ under a rigid transformation to $\mathbb R^2$.
The conics ${\cal A}(t_i)$ and ${\cal B}(t_i)$ are in contact if and only if 
the characteristic equation of ${\cal A}'(t_i)$ and ${\cal B}'(t_i)$ has a multiple root.
In any case, a candidate contact time instant that corresponds to a contact point at infinity is discarded.

\begin{center}
\begin{minipage}{0.8\textwidth}
\begin{algorithm}[H]
\caption{Computing the candidate time instants for two moving conics in
         $\mathbb R^3$}\label{alg:CCD_conics}
\begin{algorithmic}
\REQUIRE Two moving conics ${\mathcal A}(t)$ and ${\mathcal B}(t)$
         defined in the planes $\Pi_A(t)$ and $\Pi_B(t)$ in ${\mathbb R}^3$,
         $t \in [t_0,t_1]$, respectively.
\medskip
\IF {$\Pi_A(t) \equiv \Pi_B(t)$ for $t \in [t_0,t_1]$}
    \STATE Reduce CCD of ${\mathcal A}(t)$ and ${\mathcal B}(t)$ to that of
    two moving conics in a plane which is handled by Case~\caseconicsconicstwo
    (Section~\ref{sect:CCD_conics_2D})
\ELSE
    \STATE Compute the intersection line $\mathcal L(u;t)$
			     between $\Pi_A(t)$ and $\Pi_B(t)$
		\STATE Compute $h(u;t)$ and $g(u;t)$ as in Eq.~(\ref{eqn:subsL})
		       and obtain ${\hat A}(t)$ \& ${\hat B}(t)$  by rewriting
		       $h(u;t) = U^T {\hat A}(t) U $ and $g(u;t) = U^T {\hat B}(t) U$
		       where $U = (u, 1)^T$
    \STATE Compute $f(\lambda;t)=\det(\lambda {\hat A}(t)-{\hat B}(t))$	
           and the discriminant $\DeltaT$ of $f(\lambda;t)$
    \FORALL {$t_i \in \{ t \mid \Delta_f(t) = 0 \}$}
        \IF {$\Pi_A(t_i) = \Pi_B(t_i)$}
            \STATE Transform ${\cal A}(t_i)$ and ${\cal B}(t_i)$ to 
                   ${\cal A}'(t_i)$ and ${\cal B}'(t_i)$ in $\mathbb R^2$
            \IF {the characteristic equation of ${\cal A}'(t_i)$ and ${\cal B}'(t_i)$ has a multiple root}
                \STATE $\mathcal T \leftarrow \mathcal T \cup \{ t_i \}$
            \ENDIF
        \ELSIF {$\Pi_A(t_i)$ and $\Pi_B(t_i)$ are not parallel}
            \STATE $\mathcal T \leftarrow \mathcal T \cup \{ t_i \}$
        \ENDIF
    \ENDFOR
    \RETURN $\mathcal T$ as the candidate contact time instants
\ENDIF
\end{algorithmic}
\end{algorithm}
\end{minipage}
\end{center}

\subsection{Case~\caseconicsconicstwo --- conics vs. conics in $\mathbb R^2$}
\label{sect:CCD_conics_2D}

CCD of two conics in a plane is handled using the same algebraic approach as
in~\cite{ChoiWLK2006} for CCD of two ellipses in $\mathbb R^2$.
Given two moving conics ${\mathcal A}(t): {\bar X}^TA(t){\bar X}=0$ and
${\mathcal B}(t): {\bar X}^TB(t){\bar X}=0$ in $\mathbb R^2$, 
where ${\bar X} = (x, y, w)^T$ and $t \in [t_0, t_1]$,
the characteristic equation $f(\lambda;t)=\det\big(\lambda A(t)-B(t)\big)=0$
is cubic in $\lambda$.  The equation $f(\lambda;t_0)=0$
has a multiple root $\lambda_0$ if and only if
${\mathcal A}(t_0)$ and ${\mathcal B}(t_0)$ have tangential contact at time $t_0$.
Hence, we compute the discriminant $\DeltaT$ of $f(\lambda;t)$,
and the zeroes of $\DeltaT$
would be the candidate contact time instants of ${\mathcal A}(t)$ and ${\mathcal B}(t)$.
For each contact time instant $t_i$, the contact point is given by the
solution of $\big(\lambda_0 A(t)-B(t)\big){\bar X}=0$ where $\lambda_0$ 
is a multiple root of $f(\lambda;t)=0$.

\subsection{Case~\caseconicslinesthree --- conics vs. lines in $\mathbb{R}^3$}
\label{sect:CCD_conics_lines_3D}

Consider CCD between a conic $\mathcal{C}(t)$ and a line ${\mathcal L}(u;t)$ 
in $\mathbb{R}^3$.  We assume that $\mathcal{C}(t)$
is given as the intersection between a quadric $\hat{\mathcal{C}}(t): X^T {\hat C}(t) X = 0$ and a plane $\Pi_{\mathcal{C}(t)}$ in $\mathbb{R}^3$ 
(so that the axis of $\hat{\mathcal{C}}(t)$ is orthogonal to $\Pi_\mathcal{C}(t)$).
We may disregard the case of ${\mathcal L}(u;t)$ and $\Pi_{\mathcal{C}(t)}$ being always identical,
as any possible first contact of the CQMs due to $\mathcal{C}(t)$ and ${\mathcal L}(u;t)$ can then be detected by CCD of ${\mathcal L}(u;t)$ and the neighbouring boundary elements of $\mathcal{C}(t)$.
If ${\mathcal L}(u;t)$ and $\Pi_{\mathcal{C}(t)}$ are parallel for all $t$,
then ${\mathcal L}(u;t)$ and $\mathcal{C}(t)$ have no contact.
Otherwise, we obtain ${\bf p}(t)$ which is the intersection of ${\mathcal L}(u;t)$ and $\Pi_{\mathcal{C}(t)}$.  
The conic $\mathcal{C}(t_i)$ is in contact with ${\mathcal L}(u;t_i)$ in 
$\mathbb{R}^3$ at time $t_i$
if and only if ${\bf p}(t_i)$ lies on $\mathcal{C}(t_i)$, that is,
${\bf p}(t_i)^T {\hat C}(t_i) {\bf p}(t_i) = 0$, and ${\bf p}(t_i)$ is not at infinity.
Hence, the roots of ${\bf p}(t)^T {\hat C}(t) {\bf p}(t) = 0$ are the candidate contact time instants; those of which corresponding to ${\bf p}(t_i)$ at infinity are discarded.

\subsection{Case~\caseconicslinestwo --- conics vs. lines in $\mathbb{R}^2$}
\label{sect:CCD_conics_lines_2D}

Let $\mathcal{C}(t): \bar{X}^TC(t)\bar{X} = 0$ be a conic 
and $\mathcal{L}(u;t)$ be a line in $\mathbb{R}^2$,
where $\bar{X} = (x, y, w)^T \in \mathbb{PR}^2$.
By substituting $\mathcal{L}(u;t)$ into $\mathcal{C}(t)$,
we obtain $g(u;t) = L(u;t)^T C(t) L(u;t)$ which is quadratic in $u$.
Each root $t_i$ of the discriminant $\Delta_g(t)$ of $g(u;t)$ is a 
candidate contact time instant of $\mathcal{C}(t)$ and $\mathcal{L}(u;t)$,
with a corresponding contact point $\mathcal{L}(u_0, t_i)$, where
$u_0$ is a double root of $g(u;t_i)$.

\subsection{Case~\caselineslines --- lines vs. lines}
\label{sect:CCD_lines}

For CCD of two lines in $\mathbb R^3$, we seek the time instants
at which the lines intersect in $\mathbb R^3$.
Two lines $\mathcal{L}_1(u; t)={\bf p}_1(t) + u\, {\bf q}_1(t)$
and $\mathcal{L}_2(v; t)={\bf p}_2(t) + v\, {\bf q}_2(t)$ intersect in $\mathbb{PR}^3$
if and only if ${\bf q}_1(t)$, ${\bf q}_2(t)$ and
${\bf p}_2(t) - {\bf p}_1(t)$ are coplanar.  The contact time instants
are then given by
the roots of $g(t) = \det [{\bf q}_1(t), {\bf q}_2(t),
{\bf p}_2(t)-{\bf p}_1(t) ] = 0$.
The case of $g(t) \equiv 0$ is neglected since it corresponds to 
two moving lines which are always coplanar; any contact between two 
moving line segments of this kind for two CQMs can be detected by CCD of 
an end vertex of one line segment and a CQM face on which the other line segment lies.
For each candidate time instant $t_0$, the corresponding candidate contact point 
is given by ${\bf p} = {\bf p}_1(t_0) + u'\, {\bf q}_1(t_0)$, where 
$u' = \big(( {\bf p}_2(t_0) - {\bf p}_1(t_0) ) \times {\bf q}_2(t_0) \big) \cdot
      \big({\bf q}_1(t_0) \times {\bf q}_2(t_0)\big) / \big| {\bf q}_1(t_0) \times {\bf q}_2(t_0) \big|^2$ (see~\cite{Hill1994}). The straight line $\mathcal{L}_1(u; t_0)$
and $\mathcal{L}_2(u; t_0)$ are parallel and have no contact in $\mathbb{R}^3$ if 
      $\big| {\bf q}_1(t_0) \times {\bf q}_2(t_0) \big|^2 = 0$.

\section{Contact Validation}\label{sect:contact_validation}

For each contact point computed from the CCD subproblems, we need to
check if it is a valid contact point of two CQMs.
A candidate contact point $\bf p$ is a valid contact point if and only if
\begin{enumerate}
\item $\bf p$ lies on both CQMs;
\item $\bf p$ constitutes an external contact of the CQMs, which means that the interior of the CQMs does not overlap.  
\end{enumerate}

To ascertain that the first criteria is satisfied,
we assume that a CQM is obtained using CSG (constructive
solid geometry) and is represented by a CSG construction tree.
Other boundary surface representation
may entail different procedures for the validation,
but the idea is essentially the same.

Given a candidate contact point $\bf p$ between
two extended boundary elements of two CQMs ${\cal Q}_A(t_0)$
and ${\cal Q}_B(t_0)$ at time $t_0$,
let $\overset{\circ}u =$
``${\bf p}$ is in the interior of $u$'' and
$\partial u =$ ``${\bf p}$ is on the boundary of
$u$'' be two Boolean predicates. For each internal node
associated with a CSG object $w$, we will evaluate
$\overset{\circ}w$ and $\partial w$ recursively using the following
rules
(see Figure~\ref{verification}
for a 2D analogy):
\begin{align*}
\linespread{0.1} \mbox{Case } w & = u \cup v \mbox{ :} &\quad
 \overset{\circ} w \leftrightarrow{} &
 \overset{\circ}u  \vee \overset{\circ}v,& \quad
 \partial w & \leftrightarrow
  (\partial u \wedge \partial v) \vee
  (\partial u \wedge \neg \overset{\circ} v) \vee
  (\neg \overset{\circ} u \wedge \partial v)\\
\mbox{Case } w & = u \cap v \mbox{ :} &\quad
 \overset{\circ} w \leftrightarrow{} &
  \overset{\circ}u  \wedge \overset{\circ}v, & \quad
 \partial w & \leftrightarrow
  (\partial u \wedge \partial v) \vee
  (\overset{\circ}u \wedge \partial v) \vee
  (\partial u \wedge \overset{\circ}v)\\
\mbox{Case } w & = u \setminus v \mbox{ :} &\quad
 \overset{\circ} w \leftrightarrow{} &
  \overset{\circ}u \wedge \neg (\overset{\circ}v \vee \partial v),&  \quad
 \partial w & \leftrightarrow
  (\partial u \wedge \partial v) \vee
  (\partial u \wedge \neg \overset{\circ}v) \vee
  (\overset{\circ}u \wedge \partial v)
\end{align*}
It suffices to consider
the three basic Boolean operations, since all other CSG operations
can be described as their compositions.
Whether or not a point is inside or on the boundary of a quadric CSG
primitive can be checked using the quadric equation, which can be done
exactly (see Remark~\ref{rem:exact_solution}).
The answer to whether $\bf p$ is on the boundary surface of ${\cal
Q}_A(t_0)$ is then given by the truth value of the predicate
$\partial {\cal Q}_A(t_0)$ evaluated at the root node of the CSG
tree of ${\cal Q}_A(t_0)$.
Hence, a candidate contact point $\bf p$ lies on both CQMs if and only if
both $\partial {\cal Q}_A(t_0)$ and $\partial {\cal Q}_B(t_0)$
are true.

Regarding the second criteria, since the two given CQMs are separate initially,
the first occurrence of a valid contact point must guarantee an external contact of the CQMs.

\begin{figure}
\centering
\includegraphics[width=.8\linewidth]{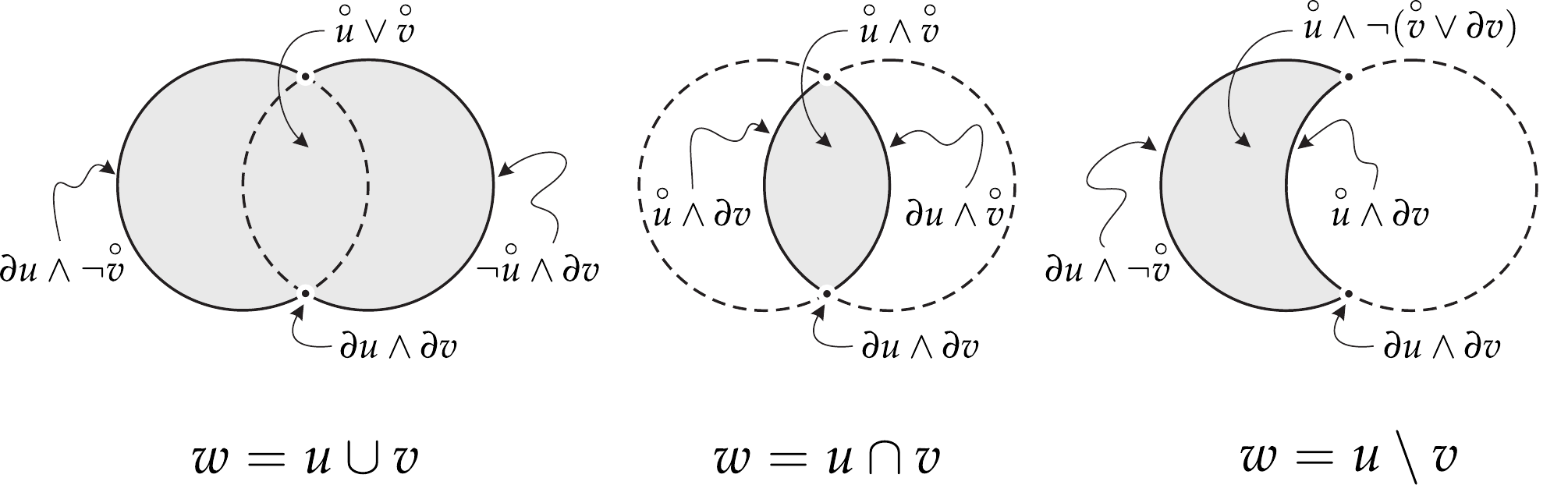}
\caption{Contact validation for CSG objects.  Given two objects $u$ and $v$,
and the result $w$ of applying a Boolean operation to $u$ and $v$,
the boundary and interior
of $w$, denoted by $\partial w$ and $\protect\overset{\circ}w$, respectively,
can be determined from $\partial u$, $\protect\overset{\circ}u$, $\partial v$ and
$\protect\overset{\circ}v$ accordingly.} \label{verification}
\end{figure}

\section{Two working examples}
\label{sect:cqm_example}

\begin{example}\label{eg:cylinder}
In this example, we solve CCD of two moving
capped elliptic cylinders $\mathcal A(t)$ and $\mathcal B(t)$,
both are of the same size (Figure~\ref{fig:cylinder_example}a).
The boundary elements of the cylinders are:
\begin{itemize}
\item Face $F_{\mathcal A,1}$, $F_{\mathcal B,1}$: a cylinder
      $\frac{x^2}{5^2}+\frac{y^2}{10^2}=1$, $z \in [-5,5]$.
\item Face $F_{\mathcal A,2}$, $F_{\mathcal B,2}$: a plane $z = -5$; and
      face $F_{\mathcal A,3}$, $F_{\mathcal B,3}$: a plane $z = 5$.
\item Edge $E_{\mathcal A,1}$, $E_{\mathcal B,1}$: an ellipse
      $\frac{x^2}{5^2}+\frac{y^2}{10^2}=1$, $z = -5$; and
      edge $E_{\mathcal A,2}$, $E_{\mathcal B,2}$: an ellipse
      $\frac{x^2}{5^2}+\frac{y^2}{10^2}=1$, $z = 5$.
\end{itemize}

\begin{figure}[!ht]
\begin{minipage}{.5\linewidth}
	\begin{center}
	\includegraphics[height=1.5in]{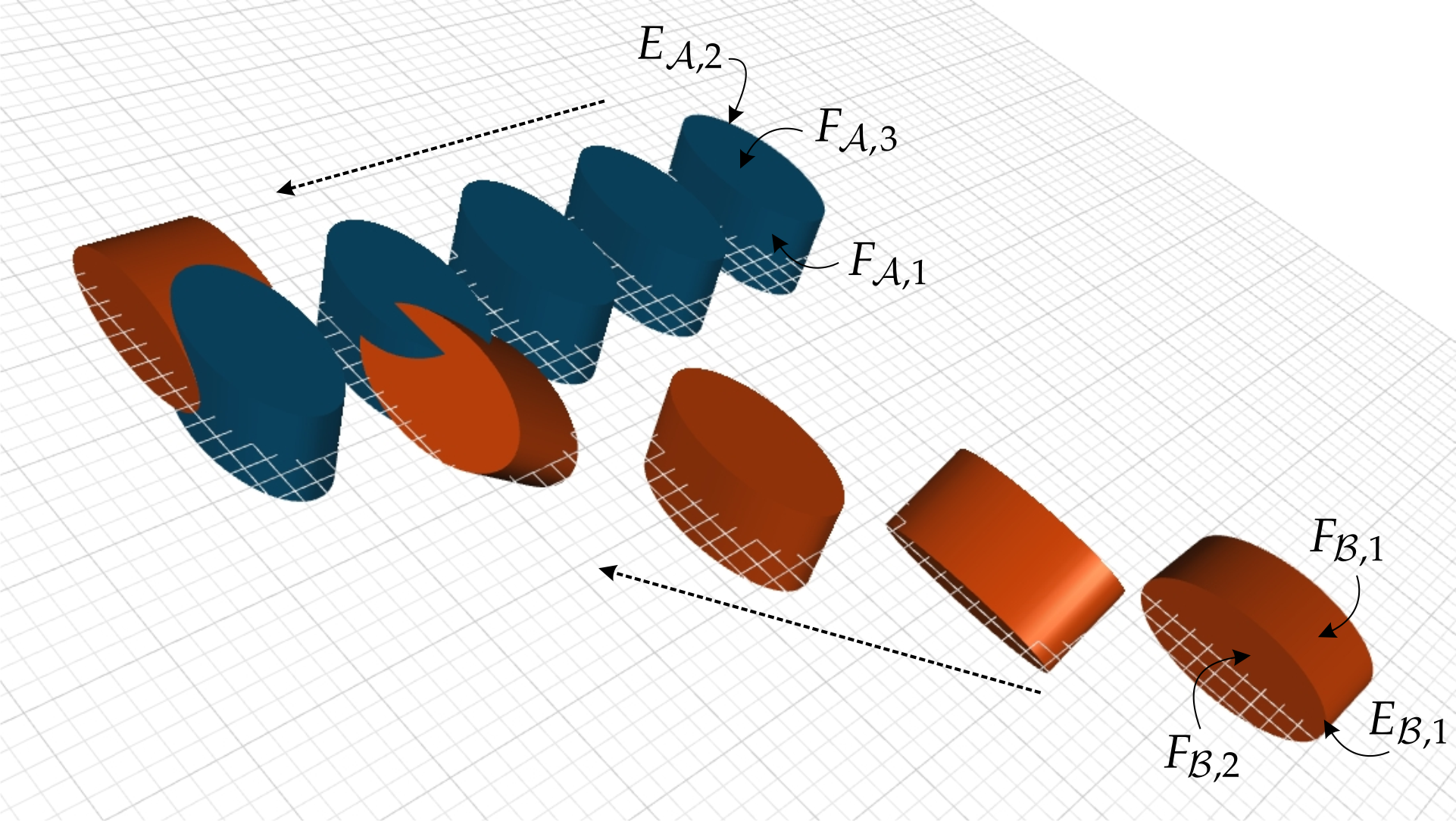}\\
	{\small (a)}
	\end{center}
\end{minipage}\hfill
\begin{minipage}{.5\linewidth}
	\begin{center}
	\includegraphics[height=1.5in]{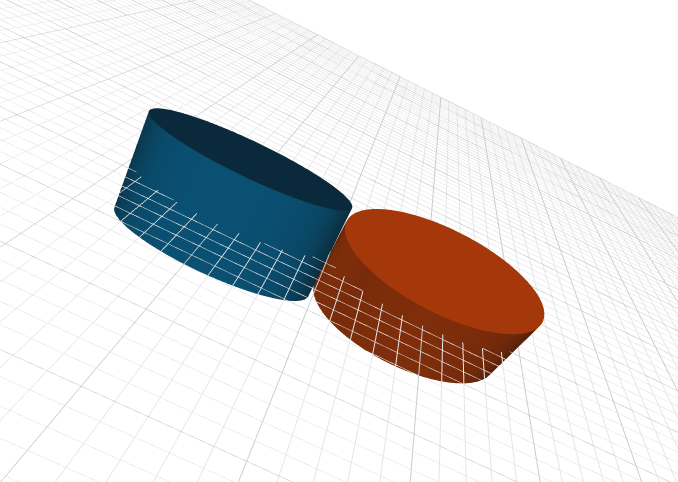}\\
	{\small (b)}
	\end{center}
\end{minipage}
\caption{(a) Two moving capped cylinders. (b) The cylinders are found
to have the first contact at $t=0.625$.}
\label{fig:cylinder_example}
\end{figure}

\label{res:motion_matrices}
Cylinder $\mathcal{A}(t)$ assumes a linear translation, while
cylinder $\mathcal{B}(t)$ assumes a degree-2 rotation as well as a linear translation.
The motion matrices of $\mathcal{A}(t)$ and $\mathcal{B}(t)$ are
{\footnotesize
\begin{equation*}
M_A(t) = \begin{pmatrix}
      1 & 0 & 0 & -60t + 30 \\
      0 & 1 & 0 & 20        \\
      0 & 0 & 1 & 0         \\
      0 & 0 & 0 & 1
         \end{pmatrix} \;\mbox{and}\;
M_B(t) = \begin{pmatrix}
	    -2t^2+2t &   0        & -2t+1    & -120t^3+180t^2-120t+30 \\
	     0       & 2t^2-2t+1  &  0       & 160t^3-260t^2+180t-50\\
	   2t-1      & 0          & -2t^2+2t & 0\\
	   0         & 0          & 0        & 2t^2-2t+1
	   	   \end{pmatrix},
\end{equation*}
}
respectively, $t \in [0,1]$.  We refer the readers to~\cite{JuettlerW2002} for the details of the construction of the rational motion $M_B(t)$.
The moving face $F_{\mathcal A,1}$ can then be expressed as $X^T A(t) X = 0$,
where $X = (x, y, z, 1)^T$ and
\label{res:quadric_equation}
\[
A(t) = M_A^{-T}(t) \begin{pmatrix} \frac{1}{5^2} & & & \\
                                   & \frac{1}{10^2} & & \\
                                   & & 0 & \\
                                   & & & -1 \end{pmatrix}
       M_A^{-1}(t).
\]
Expressions for other elements can be derived similarly by applying appropriate motion matrices.

The subproblems are listed as follows:
\begin{itemize}
\item $(F,F)$ --- $(F_{{\mathcal A},1}, F_{{\mathcal B},1})$
\item $(F,E)$ --- $(F_{{\mathcal A},1}, E_{{\mathcal B},1})$,
							$(F_{{\mathcal A},1}, E_{{\mathcal B},2})$,
							$(F_{{\mathcal B},1}, E_{{\mathcal A},1})$,
							$(F_{{\mathcal B},1}, E_{{\mathcal A},2})$,
							$(F_{{\mathcal A},2}, E_{{\mathcal B},1})$,		
							$(F_{{\mathcal A},2}, E_{{\mathcal B},2})$,
							$(F_{{\mathcal A},3}, E_{{\mathcal B},1})$,		
							$(F_{{\mathcal A},3}, E_{{\mathcal B},2})$,
							$(F_{{\mathcal B},2}, E_{{\mathcal A},1})$,		
							$(F_{{\mathcal B},2}, E_{{\mathcal A},2})$,
							$(F_{{\mathcal B},3}, E_{{\mathcal A},1})$,		
							$(F_{{\mathcal B},3}, E_{{\mathcal A},2})$
\item $(E,E)$ --- $(E_{{\mathcal A},1}, E_{{\mathcal B},1})$,
							$(E_{{\mathcal A},1}, E_{{\mathcal B},2})$,
							$(E_{{\mathcal A},2}, E_{{\mathcal B},1})$,
							$(E_{{\mathcal A},2}, E_{{\mathcal B},2})$
\end{itemize}						
We shall show how four of the above CCD subproblems
(corresponding to the four cases in Figure~\ref{capped cylinders})
is solved.  For brevity, contact point verification is skipped. 

\begin{itemize}
\item[\bf $(F,F)$:] $(F_{{\mathcal A},1}, F_{{\mathcal B},1})$---cylinder vs. cylinder

     The characteristic polynomial $f(\lambda;t) = \det(\lambda A(t)-B(t))$
     is quadratic in $\lambda$ (since $\det(A(t))\equiv0$ and $\det(B(t))\equiv0$).
     The candidate time instants are the roots of ${\rm Res}_\lambda(f, f_\lambda)=0$,
     which are found to be $t_0 = 0.5, 0.625, 0.875$.
     \begin{itemize}
     \item For $t_0 = 0.5$, we have $f(\lambda;t_0) = 0$.  
     	   The pencil $\lambda A(t_0)-B(t_0)$ is degenerate. 
     	   Hence, we transform ${\mathcal A}(t_0)$ and ${\mathcal B}(t_0)$
           by $M_A^{-1}(t_0)$ so that their axes (which are parallel) 
           are orthogonal to the $xy$-plane,
           and check whether the cross-sectional ellipses on the $xy$-plane
           have any contact.
           Now, the characteristic polynomial
           $f(\lambda;t_0) = -(16\lambda-1)(256 \lambda^2+112 \lambda+1)$
           of the cross-sectional ellipses does not have any multiple root.
           Hence, there is no contact between the cylinders at $t=0.5$. 
     \item For $t_0 = 0.625$,
           $f(\lambda;t_0)$ has a multiple root and hence the cylinders are 
           in contact.  The contact point is found to be $(-7.5, 10, 0)^T$ and 
           is verified to be a point on both capped cylinders. This is done exactly 
           (see Remark~\ref{rem:exact_solution}).
           Therefore, a valid first contact at $t=0.625$ 
           is found for the cylinders.  We may now skip
           the other larger roots of $\Delta_f(t)=0$ and also other candidate 
           time instants later than $t_0 =0.625$ obtained 
           in the subsequent CCD subproblems.  
           (Note that in the followings, calculations for the later candidate time instant 
           are still presented for illustrations, while they are 
           skipped in practice for efficiency considerations.)
%
     \end{itemize}
\item[\bf $(F,E)$:] $(F_{{\mathcal A},1}, E_{{\mathcal B},1})$---cylinder vs. ellipse

		 Both $F_{{\mathcal A},1}$ and $E_{{\mathcal B},1}$ are mapped by the same transformation such that $E_{{\mathcal B},1}$ is an ellipse in standard form on the $xy$-plane.
		 The intersection of the transformed $F_{{\mathcal A},1}$ and the $xy$-plane is an
		 ellipse ${\mathcal E}$, and CCD is performed between the two ellipses ${\mathcal E}$
		 and $E_{{\mathcal B},1}$.  The characteristic polynomial of the ellipses are found to have a double root at $t_0 = 0, 0.6341, 1$.
     \begin{itemize}
     \item For $t_0 = 0$, $f(\lambda;t_0) = 0$ has a double root $0$ which does not \label{res:reject}
           correspond to any valid contact and is hence rejected;
           ${\mathcal E}$ is indeed a line that does not touch $E_{{\mathcal B},1}$.
		 \item For $t_0 = 0.6341$, $f(\lambda;t_0) = 0$ has a double root
		 			 $\lambda_0 = -0.2843$.
		       A single contact point
		       $(-6.623, 10.414,$ $-4.95)^T$ is found, which is verified to lie on both
		       truncated cylinders.
		 \end{itemize}
\item[\bf $(F,E)$:] $(F_{{\mathcal A},3}, E_{{\mathcal B},1})$---plane vs. ellipse
		
		 Let ${\mathcal P}(t)$ be the plane $F_{{\mathcal A},3}$  and ${\mathcal E}(t)$ be the ellipse $E_{{\mathcal B},1}$.
		 Both ${\mathcal P}(t)$ and ${\mathcal E}(t)$ 
		 are simultaneously transformed such that ${\mathcal E}(t)$ is in standard 
		 form on the $xy$-plane.
		 The plane ${\mathcal P}(t)$ intersects the $xy$-plane in the line
		 ${\mathcal L}(u;t) = (10t-5, u(1-2t), 0, 4t^2-4t+1)^T$.
		 We now deal with CCD of the line ${\mathcal L}(u;t)$ and the ellipse ${\mathcal E}(t)$.
		 Substituting $\mathcal L$ into the ellipse equation yields
		 $h(u;t)$ and solving the discriminant $\Delta_h(t)$ 
		 gives the roots $t_0 = 0, 0.5, 1$.
     \begin{itemize}
     \item For $t_0 = 0$, solving $h(u;t_0)=0$ gives $u_0 = 0$, and
           the contact point is given by $X_0 = (25,-50,5)^T$.
           However, since $X_0^T E(t) X_0 = 49 > 0$, 
           $X_0$ is not in the elliptic disk on ${\mathcal P}(t)$ and hence
           $t_0 = 0$ is rejected.
     \item For $t_0 = 0.5$, ${\mathcal P}(t_0)$ is parallel to the $xy$-plane,
           and $t_0$ is therefore rejected.
     \item For $t_0 = 1$, the contact point is found to be $X_0 = (-25,30,5)^T$
           which does not lie within the elliptic disk on ${\mathcal P}(t)$ 
           and $t_0=1$ is rejected.
	\item Hence, ${\mathcal P}(t_0)$  and ${\mathcal E}(t_0)$ are collision-free.
	\end{itemize}	
\item[\bf $(E,E)$:] $(E_{{\mathcal A},1}, E_{{\mathcal B},1})$---ellipse vs. ellipse

                 Let ${\mathcal E}_A(t)$ be $E_{{\mathcal A},1}$ and ${\mathcal E}_B(t)$ be $E_{{\mathcal B},1}$.
		 We transform both ellipses simultaneously such that 
		 ${\mathcal E}_B(t)$ is in standard form 
		 on the $xy$-plane.
		 The containing planes of the ellipses are not equal 
		 for all $t$ and we proceed with CCD of two 1D ellipses,
		 and the candidate contact times are 
		 $t_0 = 0.5, 0.6342, 0.875, 0.9658$.
     \begin{itemize}
     \item For $t_0 = 0.5$, ${\mathcal E}_A(t_0)$ lies on the $xy$-plane; 
           hence, we perform collision detection for the two static ellipses
           ${\mathcal E}_1(t_0): \frac{x^2}{25}-\frac{y^2}{100}-\frac{2y}{5}+3=0$ and
           ${\mathcal E}_2(t_0): \frac{x^2}{400}+\frac{y^2}{1600}+\frac{y}{80}=0$.
           The characteristic equation for ${\mathcal E}_1(t_0)$ and
           ${\mathcal E}_2(t_0)$ has no multiple root, and hence there is no contact
           at $t_0 = 0.5$.
		 \item For $t_0 = 0.6342$, 
           the characteristic equation $f(\lambda)$ has a multiple root 
           and the contact point is found to be $(-6.7084,10.3668,-5)^T$.
		 \end{itemize}
\end{itemize}

\label{res:final_result}
{\bf Final result:}
		 Combining the results from all 17 subproblems, the two capped cylinders
		 are found to have the first contact at $t = 0.625$
		 for the pair $(F_{{\mathcal A},1}, F_{{\mathcal B},1})$ at $(-7.5, 10, 0)^T$
		 (Fig.~\ref{fig:cylinder_example}(b)).
		 The algorithm is implemented with Maple using exact algebraic computations for all CCD 
		 formulations.  We use floating point evaluation 
		 (15 significant digits) for solving the candidate time instants and computing the contact points. It takes 0.12 seconds to complete on an Intel Core 2 Duo E6600 2.40-GHz CPU (single-threaded).
\qed		 
\end{example}

\vspace*{1ex}

\begin{example}\label{eg:cqm}
In this example, we solve CCD for two moving CQMs as shown in
Figure~\ref{fig:cqm_example}(a).
Object $\cal A$ comprises 45 boundary elements (4 cylinders, 9 planes, 10 circles, 
14 lines and 8 vertices) while object $\cal B$ includes 13 boundary elements (3 cylinders, 
1 cone,  3 planes and 6 circles).
The specifications of the two objects are given in Figure~\ref{fig:cqm_spec}.
Object $\cal A$ translates linearly on the plane while $\cal B$ 
moves with a linear translation and a degree-2 rotation (Figure~\ref{fig:cqm_example}(a)).  
The motion matrices of ${\mathcal A}(t)$ and ${\mathcal B}(t)$ are
{\footnotesize
\[
M_A(t) = \begin{pmatrix}
      1 & 0 & 0 & 15-45t \\
      0 & 1 & 0 & 15-25t  \\
      0 & 0 & 1 & 0         \\
      0 & 0 & 0 & 1
         \end{pmatrix} \;\text{and} 
\]        
\[ 
\renewcommand{\arraystretch}{2.6}
M_B(t) = \begin{pmatrix}
	    \parbox{1in}{$\;(-u+1-\frtwo)t^2\\ +(u+\rtwo)t-\frtwo$} & 
	    \parbox{1in}{$\quad\quad(\frtwo-v) t^2\\   + (v-\rtwo)t+\frtwo$} &
	    -v t^2 + vt &	    \parbox{1.5in}{$(-60u+120)t^3+(90u-180)t^2\\\quad+(-30u+120)t-30$}\\
	    \parbox{1in}{$\quad\quad(v-\frtwo)t^2\\ +(\rtwo-v)t - \frtwo$} &
	    \parbox{1in}{$\quad\quad(-u-\frtwo)t^2\\ + (u+\rtwo)t - \frtwo$} &
	    (1-u)t^2+ut &
	    \parbox{1.5in}{$(9u-18)t^3+(16-8u)t^2\\\quad\quad+(-7-u)t-1$}\\
	    -vt^2+vt &
	    (u-1)t^2-ut &
	    \parbox{1in}{$\quad\quad(1-u)t^2\\\quad+(u-2)t+1$} &
	    \parbox{1.5in}{$(20u-40)t^3+(60-30u)t^2\\\quad+(-40+10u)t+10$} \\
	    0 & 0 & 0 & (2-u)t^2+(u-2)t+1
	   	   \end{pmatrix},
\]
}
respectively, $t \in [0,1]$.  
There are altogether 366 CCD subproblems, and it takes 
about 5 seconds to complete the CCD computations under the same Maple 
environment as in Example~\ref{eg:cylinder}.
The first contact configuration is found to happen at $t=0.313$ between 
the circle $\frac{(y+10)^2}{2^2}+\frac{z^2}{2^2}=1, x=-6$, of $\mathcal A$
and the cone $x^2 + z^2 = (y-8)^2, y\in[6,7]$, of $\mathcal B$ as shown in Figure~\ref{fig:cqm_example}(b).
\qed

\begin{figure}[!ht]
	\begin{center}
	\includegraphics[height=1.7in]{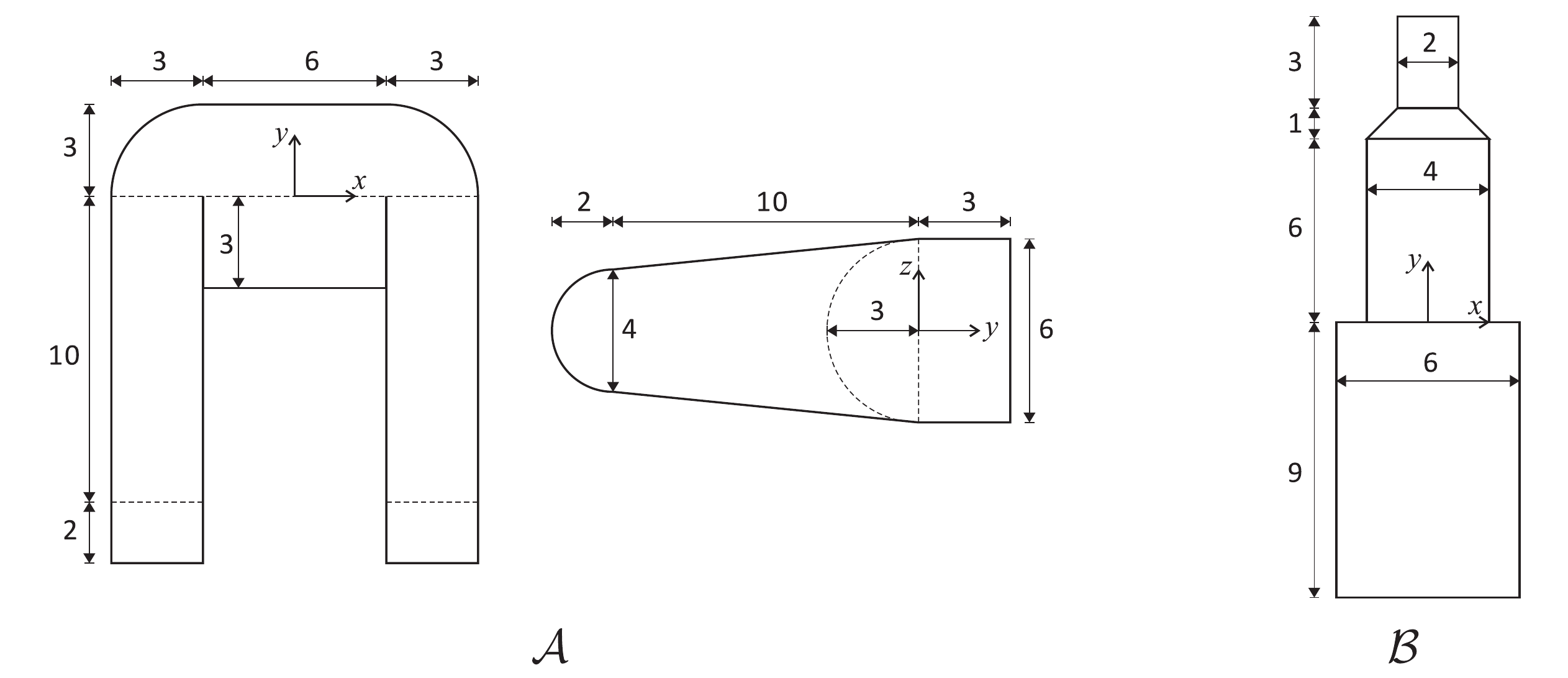}\\
	\end{center}
\caption{Specifications of two CQM objects in Example~\ref{eg:cqm}.}
\label{fig:cqm_spec}
\end{figure}

\begin{figure}[!ht]
\begin{minipage}{.5\linewidth}
	\begin{center}
	\includegraphics[height=1.7in]{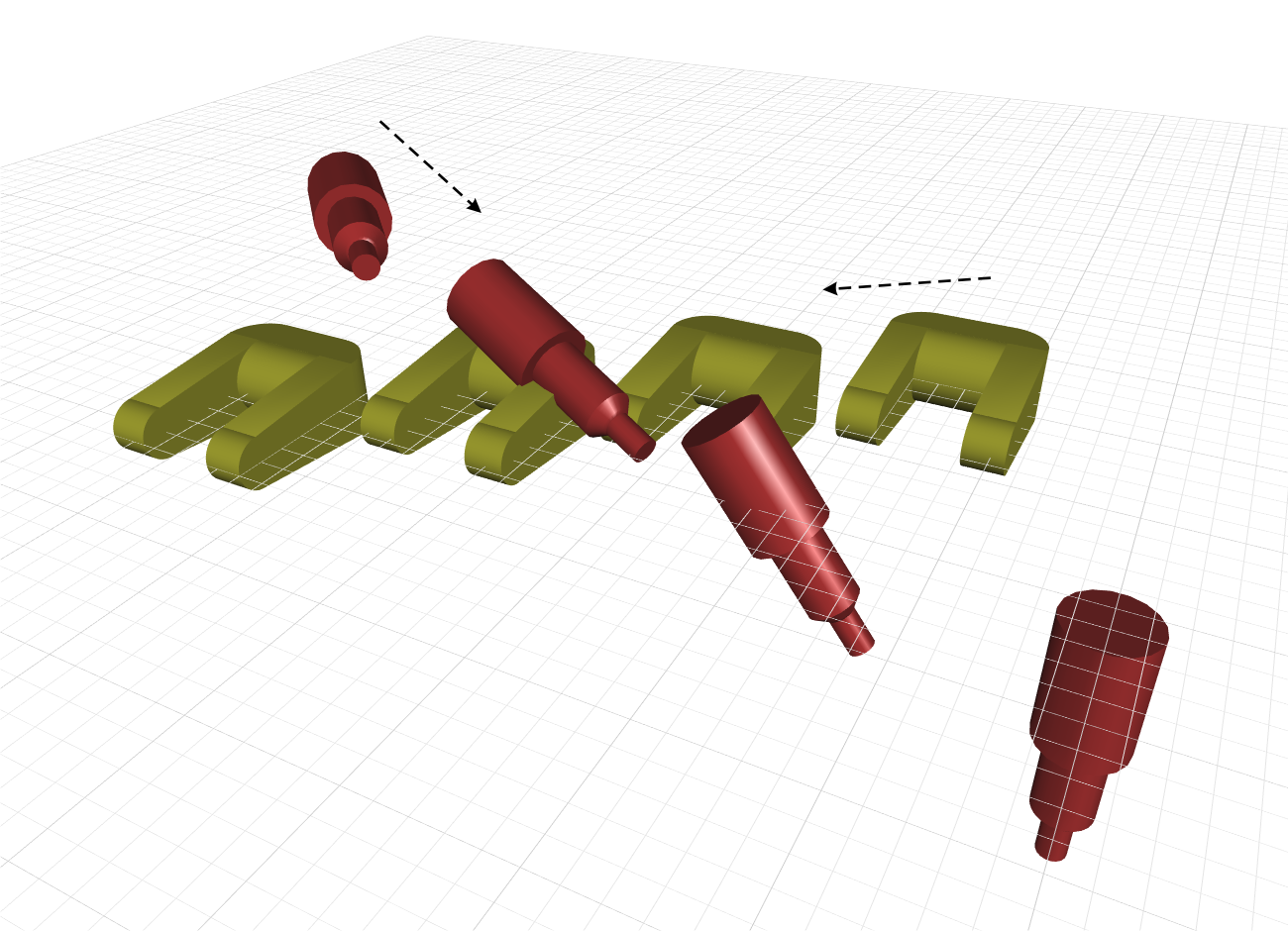}\\
	{\small (a)}
	\end{center}
\end{minipage}\hfill
\begin{minipage}{.5\linewidth}
	\begin{center}
	\includegraphics[height=1.7in]{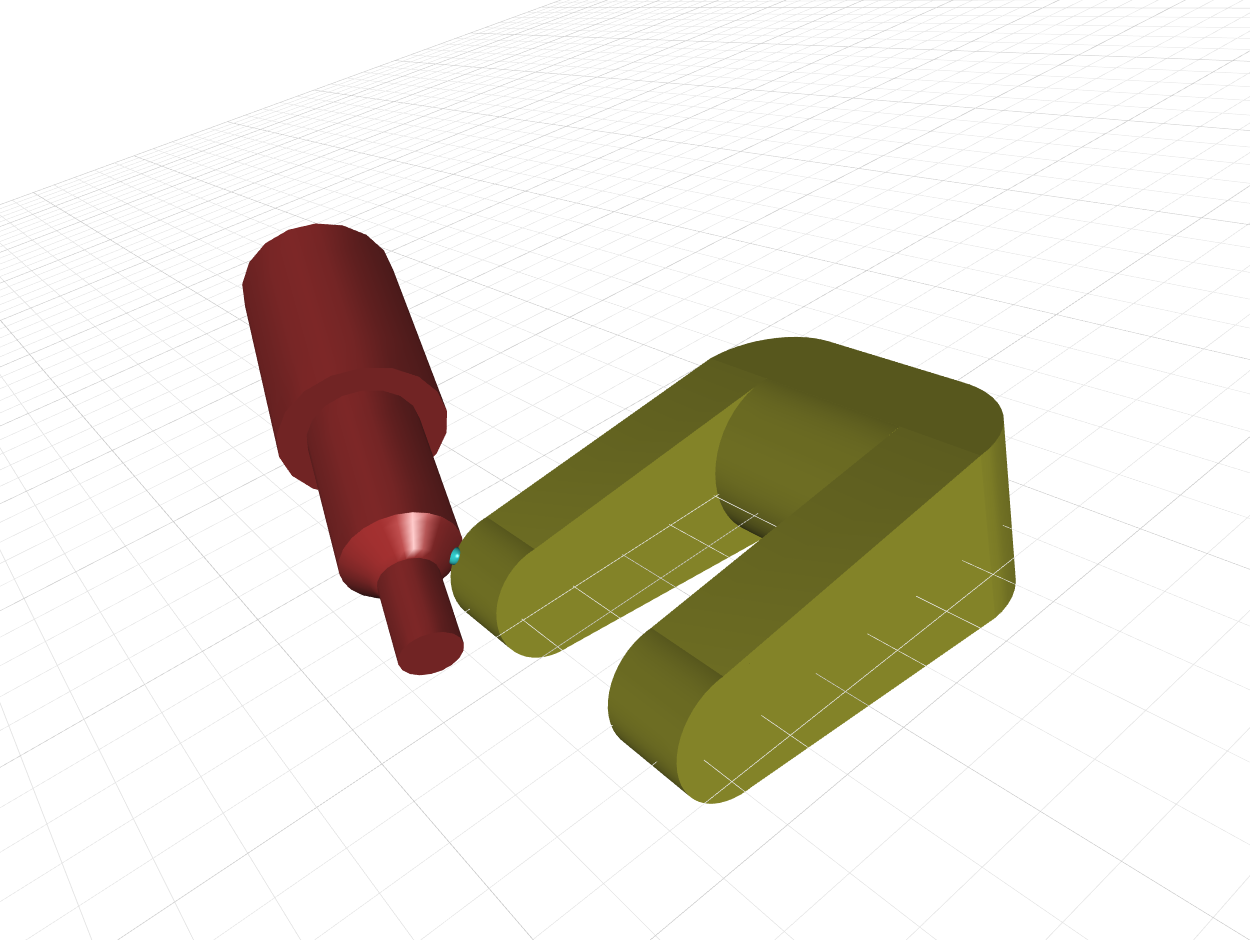}\\
	{\small (b)}
	\end{center}
\end{minipage}
\caption{(a) CCD of two CQMs in Example~\ref{eg:cqm}. (b) The first contact.}
\label{fig:cqm_example}
\end{figure}
\end{example}

\section{Conclusion}

We have presented a framework for CCD of composite
quadric models (CQMs) whose boundary surfaces are defined by
piecewise linear or quadric surface patches and whose boundary
curves are conic curves or line segments.  A hierarchy of CCD subproblems for
various types of boundary element pairs in different dimensions
are solved.  Some subproblems can be solved using a dimension
reduction technique so that the original problem is transformed
to one in a lower dimensional space.  In particular, we solved 
CCD of moving general quadrics and CCD of moving conics in 
$\mathbb R^3$.  We also developed 
procedures for contact points verification to check if
a contact point of the extended boundary elements lies on a CQM surface.

Our algorithm is exact in the sense that no approximation of the time domain 
or of the geometries is necessary.
It only requires the evaluation of polynomial expressions at real roots of other univariate polynomials, the operations of which can be performed exactly (see Remark~\ref{rem:exact_solution}).
Algebraic formulations are established for the CCD subproblems.
Out of efficiency considerations, contact time instants and the corresponding 
contact points are solved for numerically.
When the degree of motion is high, numerical stability problems thus introduced remain to be resolved.

In general, a boundary edge of a CQM may not be a conic curve,
but rather a general degree four intersection curve of two boundary
quadrics.
Algorithms for CCD of this type of general CQMs still need to be developed.
Major difficulties arise from the handling of degree four
intersection curves.
An idea is to reduce the problem of CCD of a moving general boundary edge and
a moving quadric to the study of intersection
of three quadrics in 3D (two of which intersect to give the boundary edge).
There is a contact between a quadric surface $\mathcal{A}$ and a general boundary edge which is the intersection of two quadrics $\mathcal{B}$ and $\mathcal{C}$,
if and only if $\mathcal{A}$, $\mathcal{B}$ and $\mathcal{C}$ have a common singular intersection.  The latter condition is indicated by that 
the quartic curve $G(\alpha, \beta, \gamma) \equiv \det(\alpha A +
\beta B + \gamma C) = 0$ has a singular point.
Hence, we need to develop
methods to detect the time $t_0$ at which the moving planar quartic
curve $G(\alpha, \beta, \gamma)=0$ has a singular point.
The case of CCD of two edges can be treated similarly, but is reduced
to the study of the intersection of four quadrics, that is, the two
pairs of quadrics defining two extended boundary curves. This
will then lead to the study of singularity of a quartic surface.

%
%

\bibliographystyle{plain}
\bibliography{cqm}

\begin{appendix}

\section{Proof of Lemma~\ref{lem:cqm_touch_config}}\label{app:proof}

\begin{proof}

There are at most four singular
quadrics in a nondegenerate pencil for all quadrics, since each singular quadric corresponds to 
a root of the characteristic polynomial. Hence, we may suppose that
$\mathcal{A}$ is nonsingular, for if $\mathcal{A}$ is singular we can always find a nonsingular member $\mathcal{A}'$ in the
pencil of $A$ and $\mathcal{B}$ to replace $\mathcal{A}$, and the intersection curve of $\mathcal{A}'$ and $\mathcal{B}$ is the same as that of $\mathcal{A}$ and $\mathcal{B}$. 

\noindent{\bf Case 1}:
If $rank(\lambda_0 A - B) = 3$,
$\lambda_0$ is associated with exactly one Jordan block of size $k \times k$, 
$k\geq 2$, in the Jordan canonical form of $A^{-1}B$.
\label{res:Jordan}
Then there exists a unique real eigenvector $X_0 \not=0$ (up to a scalar factor) and a generalized
eigenvector $X_1\not=0$ of $A^{-1}B$ associated with $\lambda_0$ such that
$(\lambda_0 A - B) X_0 = 0$ and  
$(\lambda_0 A - B) X_1 = X_0$.
Therefore, we have $(\lambda_0 A - B)^2 X_1 = 0$.
Then
\begin{align*}
X_0^T A X_0
&= [ (\lambda_0 A - B) X_1 ]^T A [ (\lambda_0 A - B) X_1 ]
= X_1^T (\lambda_0 A - B) A (\lambda_0 A - B) X_1
= X_1^T A (\lambda_0 A - B)^2 X_1\\
&= 0.
\end{align*}
It follows that $X_0$ is a point on $\mathcal A$.
Similarly, we can show that $X_0$ is a point on $\mathcal B$.
Since $(\lambda_0 A - B) X_0 = 0$,
we have $\lambda_0 A X_0 = B X_0$.
If $\lambda_0 \not=0$, 
$\mathcal A$ and $\mathcal B$ have a common tangent plane at $X_0$;
otherwise, $BX_0 = 0$ and the tangent plane of $\mathcal B$ at $X_0$ is undefined
which is only possible when $\mathcal B$ is a cone (since
$\mathcal B$ is irreducible) and $X_0$ is a vertex of $\mathcal B$.

\noindent{\bf Case 2}:
First we show that $\lambda_0 \not=0$.  Suppose on the contrary that
$\lambda_0=0$.  Then we have $rank(B) = rank(0\cdot A - B) = 2$ which contradicts that $rank(B) > 2$ since $B$ is irreducible.

Now, since $rank(\lambda_0 A - B) = 2$,
the null space of $A^{-1}B$ is of dimension $2$
and $\lambda_0$ is associated with two linearly independent
eigenvectors $X_0$ and $X_1$, such that
$(\lambda_0 A - B) X_0 = 0$ and
$(\lambda_0 A - B) X_1 = 0$.
Consider a line $X(u;v) = u X_0 + v X_1$.
A line in general intersects a quadric in two points (counting
multiplicity) in $\mathbb{PC}^3$, or that the line lies on the quadric.
We first assume that $X(u;v)$ intersects $\mathcal A$ at two points.
Let ${\tilde X} = X({\tilde u};{\tilde v})$ for some 
${\tilde u}, {\tilde v} \in \mathbb C$ be such an intersection point
on $\mathcal A$.
Since $(\lambda_0 A - B)X_0 = 0$ and $(\lambda_0 A - B)X_1 = 0$,
we have
\begin{align*}
0 &= {\tilde u} (\lambda_0 A - B)X_0 + {\tilde v} (\lambda_0 A - B)X_1
  = (\lambda_0 A - B){\tilde X}
  = {\tilde X}^T(\lambda_0 A - B){\tilde X}
  = \lambda_0 {\tilde X}^TA{\tilde X}
      - {\tilde X}^T B {\tilde X}\\
  &= - {\tilde X}^T B {\tilde X}.
\end{align*}

The last equality holds since ${\tilde X}$ is
on $\mathcal A$.
Hence, we have shown that
${\tilde X}$ is also on $\mathcal B$.
Moreover, since $\lambda_0 A X_0 = BX_0$ and $\lambda_0A X_1 = BX_1$,
we have
\begin{align*}
\lambda_0 A {\tilde X}
&= \lambda_0 A({\tilde u} X_0 + {\tilde v} X_1)
= {\tilde u} \lambda_0 A X_0 + {\tilde v} \lambda_0 A X_1
= {\tilde u} B X_0 + {\tilde v} B X_1\\
&= B {\tilde X}.
\end{align*}
Since $\lambda_0 \not=0$, 
the tangent planes 
$X^TA{\tilde X}=0$ and
$X^TB{\tilde X}=0$ of
$\mathcal A$ and $\mathcal B$ at ${\tilde X}$
are identical, and hence
$\mathcal A$ and $\mathcal B$ are tangential at ${\tilde X}$.

If the line $X(u;v)$ lies on the quadric surfaces,
we may show similarly that $\mathcal A$ and $\mathcal B$
are tangential at every point on $X(u;v)$.\\

\noindent{\bf Case 3}:
Again we have $\lambda_0 \not = 0$, or otherwise
$rank(B) = rank(0\cdot A - B) = 1$ which contradicts that $rank(B) > 2$ since $B$ is irreducible.

Now, since $rank(\lambda_0 A - B) = 1$,
the null space of $A^{-1}B$ is of dimension $3$
and $\lambda_0$ is associated with three linearly independent
eigenvectors $X_0$,$X_1$,$X_2$ such that
$(\lambda_0 A - B) X_0 = 0$,  
$(\lambda_0 A - B) X_1 = 0$ and
$(\lambda_0 A - B) X_2 = 0$.
Now let $X(u; v; w) = u X_0 + v X_1 + w X_2$
be the plane spanned by the eigenvectors $X_0$,$X_1$ and $X_2$.
Then this plane in general intersects $\mathcal{A}$ in a conic 
in $\mathbb{PC}^3$.
Let $\tilde X$ = $X({\tilde u}; {\tilde v}; {\tilde w})$ for
some $u', v', w' \in \mathbb R$
be a point on this conic.
We have
\begin{align*}
0 &= {\tilde u} (\lambda_0 A - B)X_0 + {\tilde v} (\lambda_0 A - B)X_1
     + {\tilde w} (\lambda_0 A - B)X_2 
  = (\lambda_0 A - B){\tilde X}
  = {\tilde X}^T
  			(\lambda_0 A - B){\tilde X}\\
  &= \lambda_0 {\tilde X}^T
  			A{\tilde X}
      - {\tilde X}^T
      	B {\tilde X}\\
  &= - {\tilde X}^T B
  			 {\tilde X}.
\end{align*}
The last equality holds as $\tilde X$
is on $\mathcal A$.
Hence, $\tilde X$ is also on $\mathcal B$.
It means that $\mathcal A$ and $\mathcal B$ share a common intersection
curve with the plane $X(u; v; w)$.
Now, we also have
\begin{align*}
\lambda_0 A {\tilde X}
&= \lambda_0 A({\tilde u} X_0 + {\tilde v} X_1 + {\tilde w}X_2)
= {\tilde u} \lambda_0 A X_0 + {\tilde v} \lambda_0 A X_1
		+ {\tilde w} \lambda_0 A X_2
= {\tilde u} B X_0 + {\tilde v} B X_1 + {\tilde w} B X_2\\
&= B {\tilde X}.
\end{align*}
Since $\lambda_0 \not= 0$, the tangent planes 
$X^T A {\tilde X}=0$ and $X^T B {\tilde X}=0$
of ${\mathcal A}$ and ${\mathcal B}$ 
are identical at every point ${\tilde X}$
on the intersection curve.
Hence, we have ${\mathcal A}$ and ${\mathcal B}$ 
are tangential along a conic curve in $\mathbb{PC}^3$.
\end{proof}

\end{appendix}

\end{document}